\documentclass[12pt,letterpaper]{article}
\usepackage[margin=2.54cm]{geometry}
\usepackage{listings}
\usepackage[dvips,final]{graphicx}
\usepackage{amsmath}
\usepackage{amsthm}
\usepackage{amsfonts}
\usepackage{dsfont}
\usepackage{epsfig}
\usepackage{braket}
\usepackage{mathtools}
\usepackage[shortlabels]{enumitem}
\usepackage{amssymb}
\usepackage{epsf}
\usepackage{epsfig}
\usepackage{ragged2e}
\usepackage{mathrsfs}
\usepackage{color}
\usepackage{appendix}
\usepackage{setspace}
\usepackage{comment}
\usepackage{enumitem}
\usepackage{tcolorbox}
\usepackage{graphicx}
\usepackage{algorithm}
\usepackage{algpseudocode}
\graphicspath{{./images/}}
\usepackage[strings]{underscore}

\usepackage{tikz}
\usetikzlibrary{shapes.geometric, shapes.misc, arrows, decorations.pathreplacing}
\tikzstyle{startstop} = [rectangle, minimum width=0.4cm, minimum height=0.4cm,text centered, draw=black, fill=red!30]
\tikzstyle{io} = [rectangle, minimum width=1cm, minimum height=0.5cm, text centered, draw=black, fill=blue!20]
\tikzstyle{process} = [rectangle, minimum width=1cm, minimum height=0.5cm, text centered, text width=2.2cm, draw=black, fill=orange!30]
\tikzstyle{blank} = [rectangle, minimum width=0.5cm, minimum height=0.5cm, text width=2.4cm, draw=white, fill=white!30]
\tikzstyle{arrow} = [thick,->,>=stealth]

\tikzstyle{inputOutput} = [rectangle, draw=white, fill=white!30]
\tikzstyle{operator} = [rectangle, draw=black, fill=white!30]
\tikzset{meter/.append style={draw, inner sep=5, rectangle, font=\vphantom{A}, minimum width=20, line width=.5,
 path picture={\draw[black] ([shift={(.1,.15)}]path picture bounding box.south west) to[bend left=60] ([shift={(-.1,.15)}]path picture bounding box.south east);\draw[black,-latex] ([shift={(0,.1)}]path picture bounding box.south) -- ([shift={(.18,-.08)}]path picture bounding box.north);}}}
\tikzset{cross/.style={cross out, draw=black, minimum size=2*(#1-\pgflinewidth), inner sep=0pt, outer sep=0pt},
cross/.default={1pt}}

\usepackage{color}
\definecolor{darkblue}{rgb}{0,0,0.5}
\definecolor{darkgreen}{rgb}{0,0.5,0}

\usepackage[
    colorlinks=true,
    linkcolor=darkblue,
    urlcolor=darkblue,
    citecolor=darkgreen
]{hyperref}
\usepackage[noabbrev]{cleveref}

\DeclareMathOperator{\tr}{Tr}

\makeatletter
\renewcommand*{\ALG@name}{Protocol}
\makeatother

\makeatletter
\providecommand\theHALG@line{\thealgorithm.\arabic{ALG@line}}
\makeatother

\makeatletter
\newenvironment{breakablealgorithm}
  {
   \begin{center}
     \refstepcounter{algorithm}
     \hrule height.8pt depth0pt \kern2pt
     \renewcommand{\caption}[2][\relax]{
       {\raggedright\textbf{\fname@algorithm~\thealgorithm} ##2\par}%
       \ifx\relax##1\relax 
         \addcontentsline{loa}{algorithm}{\protect\numberline{\thealgorithm}##2}%
       \else 
         \addcontentsline{loa}{algorithm}{\protect\numberline{\thealgorithm}##1}%
       \fi
       \kern2pt\hrule\kern2pt
     }
  }{
     \kern2pt\hrule\relax
   \end{center}
  }
\makeatother

\crefname{algorithm}{protocol}{protocols}
\Crefname{algorithm}{Protocol}{Protocols}
\crefname{line}{step}{steps}
\Crefname{line}{Step}{Steps}

\newtheorem{theorem}{Theorem}[section]

\newtheorem{lem}[theorem]{Lemma}
\newtheorem{prop}[theorem]{Proposition}

\theoremstyle{definition}
\newtheorem{defn}[theorem]{Definition}

\theoremstyle{remark}

\newcommand*{\TOGGLE}{}

\title{
Device-Independent Oblivious Transfer\\ from the Bounded-Quantum-Storage-Model \\and  Computational Assumptions}
\author{Anne Broadbent and Peter Yuen\footnote{University of Ottawa, Department of Mathematics and Statistics; \texttt{\{abroadbe,pyuen103\}@uottawa.ca}}}
\date{}

\begin{document}

\ifdefined\TOGGLE
\else
\newgeometry{left=1.91cm,right=1.91cm,top=2.54cm,bottom=2.54cm}
\fi

\maketitle
\begin{abstract}
    We present a device-independent protocol for oblivious transfer (DIOT) and analyze its security under the assumption that the receiver's quantum storage is bounded during protocol execution and that the device behaves independently and identically in each round. We additionally require that, for each device component, the input corresponding to the choice of measurement basis, and the resulting output, is communicated only with the party holding that component. Our protocol is everlastingly secure and, compared to previous DIOT protocols, it is less strict about the non-communication assumptions that are typical from protocols that use Bell inequality violations; instead, the device-independence comes from a protocol for self-testing of a single (quantum) device which makes use of a post-quantum computational assumption.\looseness=-1
\end{abstract}
\ifdefined\TOGGLE

\section{Introduction}
\label{section:intro}
\fi
Oblivious transfer (OT), described in terms of its well-known variant {\em one-out-of-two oblivious transfer} (1-2 OT), involves a sender inputting two bits, and a receiver who is then able to select and receive only one of these bits. For this two-party primitive to be secure, the sender must be unable to learn which of the two bits the receiver chose, and the receiver must be unable to learn both of the bits that were sent. OT is both intriguing and important for the fact it is universal: a secure implementation of OT enables the implementation of any cryptographic functionality between two parties \cite{Kil88}.

However, the search for an unconditionally-secure implementation of oblivious transfer has, unfortunately, been put to rest. In \cite{May97, LC97}, it was shown that, even with the power of quantum communication, it is impossible to realize an unconditionally secure implementation of bit commitment---another two-party primitive---thereby also showing the impossibility of unconditionally secure OT.

Thus, any secure OT must necessarily be realized with the aid of additional assumptions. For instance, some classical OT protocols that are secure by a post-quantum computational assumption are presented in \cite{LH19,PVW08}, though these classical protocols are not everlastingly secure, as shown in \cite{Unr13}. A different approach can be found in \cite{DFR+07}, where it is shown that the bounded-quantum-storage-model (BQSM) is sufficient for realizing secure OT. In this setting, it is assumed that the receiver's quantum storage is bounded during the execution of the protocol. Intuitively, what this enables is the following: a dishonest receiver is limited in its capacity to store qubits during the execution of the protocol, thereby forcing a measurement; this disables a large family of attacks and is formally shown to limit the probability of the adversary determining both of the sender's bits to be negligible. We say that using BQSM to realize OT in this way gives \emph{everlasting} security because a dishonest receiver who has unlimited quantum storage and processing power after the execution of the protocol does not gain any advantage. In \cite{DFR+07}, a protocol for Rand 1-2 OT\textsuperscript{$\ell$} (a \emph{randomized} variant of OT that uses $\ell$-bit strings\ifdefined\TOGGLE, described in \Cref{subsubsection:BackgroundOT}\fi) is shown to have perfect receiver-security and $\varepsilon$-sender-security. Informally, $\varepsilon$-receiver-security says that the real state of the dishonest sender is $\varepsilon$-close, in terms of trace distance, to a state that is independent of the choice bit. Similarly, $\varepsilon$-sender-security says that the real state of the dishonest receiver is $\varepsilon$-close to a state that is independent of information not indexed by the receiver's choice bit, and that information is generated randomly. When $\varepsilon=0$, the security is said to be \emph{perfect}.

The protocol for Rand~1-2~OT\textsuperscript{$\ell$} in \cite{DFR+07} assumes that the device is behaving as intended. In reality, a device may be faulty or may have been manufactured by a party with malicious intent, and the resulting behaviour can lead to a security breach in a protocol whose security proof assumes that the device behaves as intended. The assumption that a device is behaving as intended can be relaxed by lifting to the device-independent (DI) setting, where no assumptions on the inner workings of the device are made. That is, the device is treated as a black box with which the parties classically interact, and through this classical interaction alone, they can {\em test} the device to determine if it is behaving as intended.

Device-independent two-party cryptography was explored in \cite{SCA+11, AMPS16} in the form of bit commitment and coin flipping, though  it should be noted that these works necessarily focus on protocols with weak  security, since they make no additional assumptions,  and thus the impossibility of bit commitment applies. For the device-independent bounded/noisy quantum storage setting, security of the two-party primitive known as Weak String Erasure (from which OT can be constructed) has been proved in the works \cite{KW16, RTK+18}, though they do not achieve true device-independence, as it is assumed that the devices behave independently and identically in each round (this is the \emph{IID assumption)}. The authors in \cite{KW16} do prove security for the non-IID setting, though they assume that, in this case, the dishonest party only makes sequential attacks. Full device-independence (\emph{i.e.}, without the IID assumption) in two-party cryptography is, in general, quite difficult to prove. Despite this, there does exist a fully device-independent protocol for a variant of oblivious transfer, called XOR oblivious transfer, in \cite{KST22}. Soundness of the protocol in \cite{KST22} is quantified through the cheating probabilities of the dishonest parties; that is, the authors show that the cheating probability of each party can be bounded away from 1. As a result of the general difficulty with full device-independence, there exist weaker notions of device-independence such as measurement-device-independence (MDI) where only the measurement devices are treated as black boxes. Security of a protocol for MDI OT has been proved in~\cite{RW20arxiv}.

To prove security of device-independent protocols, one often relies on Bell inequalities. For instance, the Bell inequality known as the CHSH inequality was used in \cite{KW16,RTK+18} to enable device-independence, while the Mermin-Peres magic square game was used in \cite{KST22}. Simply put, the parties are able to point to the presence of an entangled quantum state in the device by observing that their input and output data violate a Bell inequality, for only by measuring an entangled quantum state is such a violation possible \cite{Bel64, CHSH69}. However, such conclusions can only hold if all loopholes have been closed. The locality loophole, for instance, describes the fact that a purely classical device, consisting of two components, can violate a Bell inequality if the two components are allowed to communicate. Thus, if one hopes to use a violation of a Bell inequality as an indicator of something quantum in the device, this loophole must obviously be closed, \emph{i.e.}, one must assume non-communication between the components.

To enforce non-communication, there are two options. The first option is to separate the two components with enough distance such that interacting with the device takes less time than a signal of light to travel from one component to the other. The second option is to perfectly shield the components of the device, thus preventing any signal from one component to reach the other. Implementation can then be done by distributing entanglement on-the-fly (\emph{i.e.}, un-shield the components after a round, distribute an EPR pair, and then re-shield for the next round). Alternatively, all Bell pairs that are to be used in the protocol could be prepared and distributed prior to the start of the protocol, though this requires that each component be able to store a large number of qubits, and it should be pointed out that OT in the BQSM precisely relies on the fact that a dishonest receiver is unable to store all qubits being used in the protocol.

\paragraph{Contributions.} We present a protocol for DI Rand 1-2 OT\textsuperscript{$\ell$} and analyze its security. Our protocol is everlastingly secure (because we work in the BQSM) and relies on the difficulty of the Learning with Errors problem (instead of a non-communication assumption).

\paragraph{Model.} As mentioned above, relying on a violation of a Bell inequality requires the non-communication assumption. For our DI Rand 1-2 OT\textsuperscript{$\ell$} protocol, we rely on a single-device self-testing protocol from \cite{MDCA21,MV21} which assumes that the device cannot solve the Learning with Errors (LWE) problem during the execution of the protocol (instead of the non-communication assumption). Note that the assumption that the LWE problem is difficult is a standard computational assumption in post-quantum cryptography \cite{Pei15eprint, Reg05}. This approach of using a post-quantum computational assumption in single-device self-testing was initiated in the groundbreaking work of \cite{Mah18} and \cite{BCM+18}. Our device is then modelled as consisting of two components connected by a quantum channel, making it equivalent to a single device. While the device can then implement any non-local operation on its two components, we shall require that our protocol has an honest implementation which only requires local operations and EPR pairs that are distributed on-the-fly (one component prepares an EPR pair and sends one qubit of the pair to the other component). We use this post-quantum computational assumption in combination with the BQSM to get everlasting security for our device-independent protocol. The exact assumptions that we make for our DI Rand 1-2 OT\textsuperscript{$\ell$} protocol are:

\begin{enumerate}
    \item {\em The quantum storage of the receiver is bounded during the execution of the protocol.}
    \item {\em The device is computationally bounded in the sense that it cannot solve the Learning with Errors (LWE) problem during the execution of the protocol.}
    \item {\em The device behaves in an IID manner, i.e., it behaves independently and identically in each round of the protocol.}
    \item {\em For each device component, the input corresponding to the choice of measurement basis, and the resulting output, is communicated only with the party holding that component.} \label{leakageAssumption}
\end{enumerate}

The first assumption is for the sake of achieving everlastingly secure OT. The second assumption is for achieving a form of device-independence that is more tolerant of communication between the device components. We believe that these first two assumptions are reasonable, as the ability to store qubits is very difficult with current technology, and it is standard to assume that the LWE problem is quantum-computationally hard. The third assumption is made for the sake of the security analysis. In our context, the fourth assumption\ifdefined\TOGGLE, discussed in greater detail in \Cref{section:DIOT}, \fi is required for retaining everlasting security. That is, because the sender and receiver do not trust each other, allowing arbitrary communication between the two components of the device makes it possible to honestly execute the protocol but use leaked information to break security at a later time. Assumption \ref{leakageAssumption} remedies this problem while being less strict than the usual non-communication assumption.\footnote{A similar problem arises in \cite{MDCA21} where the single-device self-testing protocol, which allows arbitrary quantum communication between the device components, is used to obtain a device-independent quantum key distribution (DIQKD) protocol; the problem in the DIQKD setting is that if an eavesdropper can access information sent via the communication channel, then the device could use the channel to signal to the eavesdropper. To remedy this, the authors of \cite{MDCA21} assume that the eavesdropper cannot access the communication channel.} To realize this assumption, the device components could initially be given to the sender so that they can verify the validity of the assumption; note that initially giving the components to the sender will already be required so that they can estimate the device's winning probability (see the end of \Cref{subsubsection:ModSelfTest}). When the receiver's component is returned to them, the protocol can proceed with the requirement that the receiver shields their component (as opposed to requiring that both components be shielded).

\paragraph{Methods.}

We show how to modify the Rand 1-2 OT\textsuperscript{$\ell$} protocol from \cite{DFR+07} that was proved to be secure in the BQSM and the single-device self-testing protocol for a computationally bounded device in \cite{MDCA21} to ensure their compatibility with each other. Then, in our main technical contribution, we show how the two modified protocols can be combined to create a device-independent protocol for Rand 1-2 OT\textsuperscript{$\ell$}. It should be noted that the reason the single-device self-testing protocol must be modified is that, in its form in \cite{MDCA21}, it assumes that the two parties (Alice and Bob) trust each other, and thus they can work together to test and verify the untrusted device. While this is suitable for the application to device-independent quantum key distribution (DIQKD) in \cite{MDCA21} where the two parties trust each other, it is not suitable for the setting of OT, where the sender and receiver distrust each other and the device.

\ifdefined\TOGGLE

\paragraph{Future directions.}
Proving security in the non-IID setting is an important next step. Making the IID assumption, as done in this paper, significantly simplifies the analysis, though at the cost of weakening the meaning of device-independence; for true device-independence, we must prove security in the non-IID setting. The technique of reducing the analysis of the non-IID setting to the IID setting via the entropy accumulation theorem seems like a promising approach for achieving such a task. For more on this technique, the reader is referred to \cite{DFR20, ADF+18, Arn18, ARV19}.

\paragraph{Acknowledgments.} This work is supported by Canada’s NSERC, the US Air Force Office of Scientific Research under award number FA9550-20-1-0375, and the University of Ottawa’s Research Chairs program.

\paragraph{Organization.} This paper is organized in the following manner. In \Cref{section:prelim}, definitions and technical results are given. In \Cref{section:buildingBlocks}, we introduce the two building blocks that will later be used to construct our DI Rand 1-2 OT\textsuperscript{$\ell$} protocol. In \Cref{section:DIOT}, we present \Cref{protocol:DIOTsenderVerifier}, our DI Rand 1-2 OT\textsuperscript{$\ell$} protocol, and analyze its security.

\section{Preliminaries}
\label{section:prelim}

\subsection{Notation and basics}
\label{subsection:notationBasics}

A function $n: \mathbb{N} \rightarrow \mathbb{R}_+$ is said to be {\em negligible} if $\lim_{\alpha \to \infty}n(\alpha)p(\alpha) = 0$ for any polynomial $p$. We denote an arbitrary negligible function by $\textsf{negl}(\alpha)$.

An arbitrary finite-dimensional Hilbert space is denoted as $\mathcal{H}$. A Hilbert space with complex dimension $d$ is denoted as $\mathcal{H}_d$. A density operator $\rho$ on $\mathcal{H}$ is a positive semidefinite operator on $\mathcal{H}$ with trace one, $\tr(\rho)=1$. The set of all density operators on $\mathcal{H}$ is denoted as $\mathcal{D}(\mathcal{H})$.

The \texttt{Computational} basis for $\mathcal{H}_2$ is denoted by the pair $\{\ket{0}, \ket{1}\}$. The \texttt{Hadamard} basis is defined as the pair consisting of $\ket{+}\coloneqq\frac{1}{\sqrt{2}}(\ket{0}+\ket{1})$ and $\ket{-}\coloneqq\frac{1}{\sqrt{2}}(\ket{0}-\ket{1})$. When we want to choose between the \texttt{Computational} and \texttt{Hadamard} basis  according to a bit $c \in \{0,1\}$, we shall write $[\texttt{Computational, Hadamard}]_c$ to mean
\begin{equation*}
    [\texttt{Computational, Hadamard}]_c =
    \begin{cases}
        \texttt{Computational}, &\text{if } c=0\\
        \texttt{Hadamard}, &\text{if } c=1.
    \end{cases}
\end{equation*}

With bits $v^\alpha,v^\beta \in \{0,1\}$, we can denote the four Bell states (also called Bell pairs) by
\begin{equation}
\label{bellPairs}
    \ket{\phi^{(v^\alpha, \, v^\beta)}} = (\sigma^{v^\alpha}_Z \, \sigma^{v^\beta}_X \otimes \mathds{1}) \frac{\ket{00}+\ket{11}}{\sqrt{2}},
\end{equation}
where
\begin{equation*}
    \sigma_X =
    \begin{pmatrix}
    0 & 1\\
    1 & 0
    \end{pmatrix}
    \quad \text{ and } \quad
    \sigma_Z =
    \begin{pmatrix}
    1 & 0\\
    0 & -1
    \end{pmatrix}
\end{equation*}
are the single-qubit Pauli operators. If $x$ and $y$ are the bases used to measure, respectively, the first and second qubit of the state $\ket{\phi^{(v^\alpha, \, v^\beta)}}$, with respective outcomes $a$ and $b$, then the following holds:
\begin{equation}
\label{equalBits}
\begin{cases}
    a \oplus b = v^\beta, \quad \text{if } x = y =\texttt{Computational}\\
    a \oplus b = v^\alpha, \quad \text{if } x = y =\texttt{Hadamard}
\end{cases}
\end{equation}

Denote by $\{Q^a_x\}_{a \in \{0,1\}}$ the single-qubit measurement in the basis $x$ (that is, for \newline $x=\texttt{Hadamard}$, $Q^0_x=\ket{+}\bra{+}$ and $Q^1_x = \ket{-}\bra{-}$).

The trace distance between two density operators $\rho,\sigma$ is defined by $D(\rho,\sigma) \coloneqq \frac{1}{2}\tr|\rho-\sigma|$, where we define $|A| \coloneqq \sqrt{A^*A}$ to be the positive square root. The notation $\rho \approx_\varepsilon \sigma$ means that $\rho$ and $\sigma$ are $\varepsilon$-close: $D(\rho,\sigma) \leq \varepsilon$.

A quantum state in a register $E$ is fully described by its density matrix $\rho_E$. In some cases, a quantum state may depend on some classical random variable $X$ in the sense that it is described by the density matrix $\rho^x_E$ if and only if $X=x$. If an observer only has access to $E$, but not $X$, then the state is determined by the density matrix
\begin{equation*}
    \sum_x P_X(x) \rho^x_E.
\end{equation*}
The joint state, consisting of both the classical $X$ and quantum register $E$, is called a \emph{cq-state} and is described by the density matrix
\begin{equation*}
    \sum_x P_X(x) \ket{x}\bra{x} \otimes \rho^x_E.
\end{equation*}
The following notation is then used:
\begin{equation*}
    \rho_{XE} = \sum_x P_X(x) \ket{x}\bra{x} \otimes \rho^x_E
    \quad \text{ and } \quad
    \rho_E = \tr_X (\rho_{XE}) = \sum_x P_X(x)\rho^x_E.
\end{equation*}
The quantum representation of a classical random variable $X$ is $\rho_X = \Sigma_x P_X(x) \ket{x}\bra{x}$. The above notation extends to quantum states that depend on more than one classical random variable, leading to \emph{ccq-states}, \emph{cccq-states}, and so on.

Note that $\rho_{XE} = \rho_X \otimes \rho_E$ holds if and only if the quantum part is independent of~$X$. Additionally, if $\rho_{XE}$ and $\rho_X \otimes \rho_E$ are $\varepsilon$-close in terms of trace distance, then the real system~$\rho_{XE}$ behaves as the ideal system $\rho_X \otimes \rho_E$ except with probability $\varepsilon$.


\subsection{Smooth R\'{e}nyi entropy}
\label{subsection:entropy}

\begin{defn}[{\bf R\'{e}nyi entropy} \cite{Ren61}]
Let $P$ be a probability distribution over a finite set $\mathcal{X}$ and $\alpha \in [0, \infty]$. Then the {\em R\'{e}nyi entropy} of order $\alpha$ is defined as
\begin{equation*}
    H_{\alpha}(P) = -\log\bigg(\big(\sum_{x \in \mathcal{X}} P(x)^\alpha\big)^{\frac{1}{\alpha-1}}\bigg)
\end{equation*}
\end{defn}

In the limit $\alpha \rightarrow \infty$, we get the {\em min-entropy} $H_{\infty}(P)=-\log(\max_{x \in \mathcal{X}}P(x))$ and in the limit $\alpha \rightarrow 0$, we get the {\em max-entropy} $H_0(P) = \log |\{x \in \mathcal{X} : P(x)>0\}|$.

Now, for the case where we have a random variable $X$ with probability distribution $P_X$, we introduce a slight abuse of notation by writing $H_{\alpha}(X)$ instead of $H_{\alpha}(P_X)$. Keeping this in mind, the conditional min-entropy and max-entropy is, respectively, defined as
\begin{align*}
    H_{\infty}(X|Y) &\coloneqq \min_{y} H_{\infty}(X|Y=y) = \min_{x,y} \bigg(-\log P_{X|Y=y}(x)\bigg)\\
    H_0(X|Y) &\coloneqq \max_y H_0(X|Y=y) = \max_y \log|\{x \in \mathcal{X} : P_{X|Y=y}(x)>0\}|
\end{align*}

Related to R\'{e}nyi entropy is the notion of {\em smooth R\'{e}nyi entropy}, which was introduced in \cite{RW04}. In \cite{Ren05} and \cite{RW05}, smooth R\'{e}nyi entropy is used to generalize the notion of conditional min- and max-entropy to {\em conditional smooth min-} and {\em max-entropy} which are, respectively, defined as
\begin{align*}
    H^{\varepsilon}_{\infty}(X|Y) &\coloneqq \max_\mathcal{E} \min_{x,y} \bigg(-\log \bigg(\frac{P_{XY\mathcal{E}}(x,y)}{P_Y(y)}\bigg)\bigg)\\
    H^{\varepsilon}_0(X|Y) &\coloneqq \min_{\mathcal{E}} \max_y \log |\{x \in \mathcal{X} : \frac{P_{XY\mathcal{E}}(x,y)}{P_Y(y)}>0\}|
\end{align*}
where the maximum/minimum ranges over all events $\mathcal{E}$ with probability $\text{Pr}[\mathcal{E}] \geq 1-\varepsilon$, and $P_{XY\mathcal{E}}(x,y)$ is the probability that $\mathcal{E}$ occurs and $X,Y$ take values $x,y$, respectively.

We now list some results which we shall need to use in later sections.

\begin{lem}[{\bf Chain rule} \cite{RW05}]
\label{chainRule}
\begin{equation*}
    H^{\varepsilon+\varepsilon'}_{\infty}(X|Y) > H^{\varepsilon}_{\infty}(XY)- H_0(Y) - \log\bigg(\frac{1}{\varepsilon'}\bigg)
\end{equation*}
for all $\varepsilon,\varepsilon' >0$.
\end{lem}

The following lemma is a special case of a more general entropic uncertainty relation introduced in \cite{DFR+07}.

\begin{lem}[{\cite{DFR+07}}]
\label{entropicUncertainty}
Let $\rho \in \mathcal{D}(\mathcal{H}^{\otimes n}_2)$ be an arbitrary $n$-qubit quantum state. Let $\Theta=(\Theta_1, \ldots, \Theta_n)$ be uniformly distributed over $\{\texttt{Computational, Hadamard}\}^n$, and let $X=(X_1, \ldots X_n)$ be the outcome when measuring $\rho$ in basis $\Theta$. Then for any $0 < \lambda < 1/2$,
\begin{equation*}
    H^{\varepsilon}_{\infty}(X|\Theta) \geq \bigg(\frac{1}{2}-2\lambda\bigg)n
\end{equation*}
where $\varepsilon=\exp(-\frac{\lambda^2 n}{32(2-\log(\lambda))^2})$.
\end{lem}

The following lemma is presented as a corollary in \cite{DFR+07} to the Min-Entropy-Splitting lemma---which essentially states that if the joint entropy of two random variables is large, then at least one of the random variables must have at least half of the original entropy in a randomized sense.

\begin{lem}[\cite{DFR+07}]
\label{entropySplit}
Let $\varepsilon \geq 0$, and let $X_0, X_1, Z$ be random variables (over possibly different alphabets) such that $H^{\varepsilon}_{\infty}(X_0 X_1 | Z) \geq \alpha$. Then, there exists a binary random variable~$C$ over $\{0,1\}$ such that $H^{\varepsilon+\varepsilon'}_{\infty}(X_{1-C}|ZC) \geq \alpha/2-1-\log(1/\varepsilon')$ for any $\varepsilon' > 0$.
\end{lem}

Let $\mathcal{F}$ be a class of hash functions from $\{0,1\}^n$ to $\{0,1\}^\ell$. The class $\mathcal{F}$ is said to be {\em two-universal} \cite{CW79} if, for any distinct $x,x' \in \{0,1\}^n$ and for~$F$ uniformly distributed over~$\mathcal{F}$, it holds that $\text{Pr}[F(x)=F(x')] \leq 1/2^\ell$.

\begin{theorem}[{\bf Privacy amplification} \cite{DFR+07}]
\label{privacyAmp}
Let $\varepsilon \geq 0$. Let $\rho_{XUE}$ be a ccq-state, where X takes values in $\{0,1\}^n$, U in the finite domain $\mathcal{U}$ and register $E$ contains $q$ qubits. Let $F$ be the random and independent choice of a member of a two-universal class of hash functions $\mathcal{F}$ from $\{0,1\}^n$ into $\{0,1\}^{\ell}$. Then
\begin{equation*}
    D(\rho_{F(X)FUE}, \tfrac{1}{2^{\ell}}\mathds{1} \otimes \rho_{FUE}) \leq \frac{1}{2} 2^{-\frac{1}{2}(H^{\varepsilon}_{\infty}(X|U)-q-\ell)}+2\varepsilon.
\end{equation*}
\end{theorem}


\subsection{Extended noisy trapdoor claw-free function family}
\label{subsection:ENTCF}

In \Cref{subsection:DIcomp}, we shall introduce the single-device self-testing protocol that we will be utilizing. The underlying cryptographic primitive of this protocol is an extended noisy trapdoor claw-free function family (ENTCF family) \cite{BCM+18, Mah18}. The Learning with Errors problem \cite{Reg05} can be used to construct ENTCF families; to see this, the reader is referred to section 4 of \cite{BCM+18} and section 9 of \cite{Mah18}.

We now present an informal description, taken from \cite{MV21, Mah18}, of the properties of ENTCF families. The formal description of these properties can be found in \cite{Mah18}.

An ENTCF family consists of two families of function pairs: $\mathcal{F}$ and $\mathcal{G}$. It is assumed that both the common domain of all function pairs and the codomain are sets of all bit strings of a fixed length. A function pair $(f_{k,0}, f_{k,1})$ is indexed by a public key $k$. If $(f_{k,0},f_{k,1}) \in \mathcal{F}$, then it is called a {\em  claw-free pair}; and if $(f_{k,0},f_{k,1}) \in \mathcal{G}$, then it is called an {\em injective pair}.
\begin{enumerate}[label=(\roman*)]
    \item For a fixed $k \in \mathcal{K_F}$, $f_{k,0}$ and $f_{k,1}$ are bijections with the same image; so, for every $y$ in their image, there exists a unique pair of pre-images $(x_0,x_1)$, called a {\em claw}, such that
    \begin{equation*}
        f_{k,0}(x_0)=f_{k,1}(x_1)=y.
    \end{equation*}
    \item Given a key $k \in \mathcal{K_F}$ for a claw-free pair, it is quantum-computationally hard to find a claw for the corresponding function pair; but with access to a $y$ in the image and a trapdoor, finding the claw is computationally easy; this is the claw-free property. A stronger version of this property is the \emph{adaptive hardcore bit property}: without access to a trapdoor, it is quantum-computationally hard to simultaneously compute a preimage $x_i$ and a bit of the form $d \cdot (x_0 \oplus x_1)$ for any non-trivial bit string $d$, where $\cdot$ here is the inner product between bit strings, and where $(x_0,x_1)$ forms a valid claw.\label{adaptHardBit}
    \item For a fixed $k \in \mathcal{K_G}$, $f_{k,0}$ and $f_{k,1}$ are injective functions with disjoint images. With a trapdoor and $y=f_{k,b}(x_{b,y})$, where $b\in\{0,1\}$, one can efficiently recover $(b,x_{b,y})$.
    \item Given a key $k \in \mathcal{K_F} \cup \mathcal{K_G}$, it is quantum-computationally hard (without access to trapdoor information) to determine whether $k$ is a key for a claw-free or an injective pair (\emph{i.e.} it is computationally indistinguishable from a key from $\mathcal{K_F}$). This is known as {\em injective invariance}.\label{injInv}
    \item For every $k \in \mathcal{K_F} \cup \mathcal{K_G}$, there exists a trapdoor $t_k$ which can be sampled together with~$k$ and with which (ii) and (iv) are computationally easy.\label{trapdoorEasy}
\end{enumerate}


\section{Building Blocks}
\label{section:buildingBlocks}

In this section, we introduce the two building blocks that will be used to create our DI Rand~1-2 OT\textsuperscript{$\ell$} protocol in \Cref{section:DIOT}.

\Cref{subsection:OTinBQSM} comprises the first building block. In \Cref{subsubsection:BackgroundOT}, we describe Rand~1-2~OT\textsuperscript{$\ell$} and present a security definition that comes almost directly from \cite{DFR+07}, though it is modified to account for protocol aborts. Then in \Cref{subsubsection:ModRandOT}, we present \Cref{protocol:RandQOT}, which is our modified version of the protocol in \cite{DFR+07} for Rand 1-2 OT\textsuperscript{$\ell$}.

\Cref{subsection:DIcomp} comprises the second building block. In \Cref{subsubsection:OgSelfTest}, we present the single-device self-testing protocol from \cite{MDCA21} and outline how it works. Then in \Cref{subsubsection:ModSelfTest}, we present \Cref{protocol:testOneVerifier}, which is merely a minor modification of the single-device self-testing protocol.

\subsection{Oblivious transfer in the bounded-quantum-storage model}
\label{subsection:OTinBQSM}

\subsubsection[Background of Rand 1-2 OT]{Background of Rand 1-2 OT\textsuperscript{$\ell$}}
\label{subsubsection:BackgroundOT}

We follow \cite{DFR+07} in discussing oblivious transfer in the bounded-quantum-storage model. First, consider 1-out-of-2 oblivious transfer (1-2 OT\textsuperscript{$\ell$}) which is described in the context of two parties: the sender and the receiver. The sender sends two $\ell$-bit strings $S_0$ and $S_1$ to the receiver such that the receiver can choose which of the two strings they want to receive, but they do not get to learn anything about the other string, and the sender does not get to learn which string the receiver has chosen.

We will be concerned with Rand 1-2 OT\textsuperscript{$\ell$} which operates in the same manner as 1-2~OT\textsuperscript{$\ell$} except that the two strings $S_0$ and $S_1$ are generated uniformly at random during the protocol and output to the sender, as opposed to the sender inputting the strings. From Rand 1-2~OT\textsuperscript{$\ell$}, one can then construct 1-2 OT\textsuperscript{$\ell$} \cite{DFSS06}.

Let us now formally define what a Rand 1-2 OT\textsuperscript{$\ell$} protocol is and what security for such a protocol means. This definition (\Cref{randOT}) comes almost directly from \cite{DFR+07}; the only change we have made is the inclusion of a way to account for protocol aborts.\footnote{The authors would like to thank an anonymous reviewer for suggesting a way of handling aborts in our main protocol.} Regarding notation in this definition, an honest sender is denoted as $S$ and a dishonest sender as~$\widetilde{S}$. Similarly, $R$ for an honest receiver and $\widetilde{R}$ for a dishonest receiver. Additionally, $C$ denotes the binary random variable which describes $R$'s choice bit; $S_0, S_1$ denote the $\ell$-bit long random variables describing $S$'s output strings; $Y$ denotes the $\ell$-bit long random variable describing $R$'s output string (which should be $S_C$ when both are honest); and $Z$ denotes the binary random variable which describes whether the protocol aborts ($Z=1$ corresponds to the protocol not aborting). Given a protocol for Rand 1-2 OT\textsuperscript{$\ell$}, the overall quantum state in the case of a dishonest sender $\widetilde{S}$ is given by the ccq-state $\rho_{CY\widetilde{S}}$, and in the case of a dishonest receiver $\widetilde{R}$, the overall quantum state is given by the ccq-state $\rho_{S_0S_1\widetilde{R}}$.

\begin{defn}[{\bf Rand 1-2 OT\textsuperscript{$\ell$}}]
\label{randOT}
An $\varepsilon$-secure Rand 1-2 OT\textsuperscript{$\ell$} is a quantum protocol between a sender $S$ and a receiver $R$. The sender has no input but the receiver has input $C \in \{0,1\}$ such that for any distribution of $C$, if $S$ and $R$ follow the protocol, then $S$ gets output $S_0, S_1 \in \{0,1\}^{\ell}$ and $R$ gets $Y=S_C$, except with probability $\varepsilon$, and, with $Z \in \{0,1\}$ describing whether the protocol aborts ($Z=0$) or does not abort ($Z=1$), the following two properties hold:

\vspace{2mm}
\textbf{$\varepsilon$-Receiver-security}: If $R$ is honest, then for any $\widetilde{S}$, there exist random variables $S_0',S_1'$ such that $\text{Pr}[Y=S_C'] \geq 1-\varepsilon$ and, with $\rho^{Z=1}_{CS_0'S_1'\widetilde{S}}$ being the state when the protocol does not abort,
\begin{equation}
\label{receiverSecurity}
    \Pr(Z=1) D(\rho^{Z=1}_{C S_0' S_1' \widetilde{S}}, \rho^{Z=1}_C \otimes \rho^{Z=1}_{S_0' S_1' \widetilde{S}}) \leq \varepsilon.
\end{equation}

\textbf{$\varepsilon$-Sender-security}: If $S$ is honest, then for any $\widetilde{R}$, there exists a random variable $C'$ such that, with $\rho^{Z=1}_{S_{1-C'}S_{C'}C'\widetilde{R}}$ being the state when the protocol does not abort,
\begin{equation}
\label{senderSecurity}
    \Pr(Z=1) D\big(\rho^{Z=1}_{S_{1-C'} S_{C'} C' \widetilde{R}}, \frac{1}{2^{\ell}} \mathds{1}  \otimes \rho^{Z=1}_{S_{C'} C' \widetilde{R}}\big) \leq \varepsilon.
\end{equation}

\textbf{Completeness}: When both parties and the device are honest, the probability of aborting is small,
\begin{equation}
    \Pr(Z=0) \leq \varepsilon.
\end{equation}

Additionally, if any of the above holds for $\varepsilon=0$, then the corresponding property is said to hold perfectly. If one of the properties only holds with respect to a restricted class $\mathfrak{S}$ of $\widetilde{S}$'s (respectively $\mathfrak{R}$ of $\widetilde{R}$'s), then this property is said to hold and the protocol is said to be secure against $\mathfrak{S}$ (respectively $\mathfrak{R}$).
\end{defn}

Receiver-security says that whatever the actions of a dishonest sender, they are just as good as the following actions: cause the protocol to abort or generate the ccq-state $\rho_{S_0' S_1' \widetilde{S}}$ independently of $C$, inform $R$ of~$S'_C$, and output $\rho_{\widetilde{S}}$. Sender-security says that whatever the actions of a dishonest receiver, they are just as good as causing the protocol to abort or having a ccq-state $\rho_{S_{C'} C' \widetilde{R}}$, informing $S$ of~$S_{C'}$ and an independent uniformly distributed $S_{1-C'}$, and outputting $\rho_{\widetilde{R}}$.

A protocol satisfying \Cref{randOT} is then a secure implementation of Rand 1-2 OT\textsuperscript{$\ell$} with the exception that a dishonest sender may influence the distribution of $S_0$ and $S_1$, and a dishonest receiver may influence the distribution of the string of their choice, though as pointed out in \cite{DFR+07}, this is acceptable for a straightforward construction of the standard 1-2 OT\textsuperscript{$\ell$}.
For an explanation of why the existence of $S_0'$ and $S_1'$ in the statement of receiver-security is necessary, the reader is referred to the discussion after Definition 3 in \cite{DFR+07}.

\subsubsection[Modified Rand 1-2 OT protocol]{Modified Rand 1-2 OT\textsuperscript{$\ell$} protocol}
\label{subsubsection:ModRandOT}

Let us now present a slightly modified version of a quantum protocol for Rand 1-2 OT\textsuperscript{$\ell$} from \cite{DFR+07}. This modified protocol, given as \Cref{protocol:RandQOT} below, differs from the Rand 1-2 OT\textsuperscript{$\ell$} protocol in \cite{DFR+07} in that their protocol has the sender use conjugate coding to send quantum states to the receiver; they then show that to prove sender security, it suffices to prove sender security for an EPR-based version of their protocol, and thus do so for receivers with bounded quantum storage. Our protocol, however, explicitly uses Bell pairs; furthermore, our protocol allows for the use any of the four possible Bell pairs while the EPR-based version of the protocol in \cite{DFR+07} only needed to use the $\ket{\phi^{(0,0)}}$ Bell pair.

These modifications are minor and so the proofs in \cite{DFR+07} for perfect receiver-security and $\varepsilon$-sender-security against quantum-memory-bounded receivers largely carries over to our protocol. Nevertheless, for the sake of being explicit, we shall give the proofs here and note how our modifications affect them.

In \Cref{protocol:RandQOT}, we let $\mathcal{F}$ be a fixed two-universal class of hash functions from $\{0,1\}^n$ to $\{0,1\}^\ell$, where $\ell$ is to be determined later. Note that we can apply a function $f \in \mathcal{F}$ to a $n'$-bit string with $n' < n$ by padding it with zeros; so, for instance, if $x|_I$ is an $n'$-bit string and we write $f(x|_I)$, then it will be assumed that $x|_I$ has been padded with zeros to form an $n$-bit long string.

\begin{breakablealgorithm}
\caption{Rand 1-2 OT\textsuperscript{$\ell$} with Bell pairs}
\label{protocol:RandQOT}
\begin{algorithmic}[1]

\State
\label{protocol:RandQOT:step1}
A device prepares $n$ uniformly random Bell pairs $\ket{\phi^{(v^{\alpha}_i, \, v^{\beta}_i)}}$, $i=1, \ldots, n$, where the first qubit of each pair goes to $S$ along with the string $v^{\alpha}$, and the second qubit of each pair goes to $R$ along with the string $v^{\beta}$.

\vspace{2mm}
\State
\label{protocol:RandQOT:step2}
$R$ measures all qubits in the basis $y=[\texttt{Computational, Hadamard}]_c$ where $c$ is $R$'s choice bit. Let $b \in \{0,1\}^n$ be the outcome. $R$ then computes $b \oplus w^\beta$, where the $i$-th entry of $w^\beta$ is defined by
\begin{equation*}
    w^\beta_i \coloneqq
    \begin{cases}
        0, \quad &\text{if } y = \texttt{Hadamard}\\
        v^\beta_i, \quad &\text{if } y = \texttt{Computational}
    \end{cases}
\end{equation*}

\vspace{2mm}
\State
\label{protocol:RandQOT:step3}
$S$ picks uniformly random $x \in \{\texttt{Computational, Hadamard}\}^n$, and measures the $i$-th qubit in basis $x_i$. Let $a \in \{0,1\}^n$ be the outcome. $S$ then computes $a \oplus w^\alpha$, where the $i$-th entry of $w^\alpha$ is defined by
\begin{equation*}
    w^\alpha_i \coloneqq
    \begin{cases}
        v^\alpha_i, \quad &\text{if } x_i = \texttt{Hadamard}\\
        0, \quad &\text{if } x_i = \texttt{Computational}
    \end{cases}
\end{equation*}

\vspace{2mm}
\State
\label{protocol:RandQOT:step4}
$S$ picks two uniformly random hash functions $f_0,f_1 \in \mathcal{F}$, announces $x$ and $f_0, f_1$ to~$R$, and outputs $s_0 \coloneqq f_0 (a \oplus w^\alpha|_{I_0})$ and $s_1 \coloneqq f_1(a \oplus w^\alpha|_{I_1})$ where \newline $I_r \coloneqq \{i:x_i = [\texttt{Computational, Hadamard}]_r\}$.

\vspace{2mm}
\State
\label{protocol:RandQOT:step5}
$R$ outputs $s_c = f_c(b \oplus w^\beta|_{I_c})$.

\end{algorithmic}
\end{breakablealgorithm}

Observe that without the bounded quantum storage assumption, the receiver can easily break the security of \Cref{protocol:RandQOT} by choosing to not measure their qubits in \cref{protocol:RandQOT:step2} and instead store them until \cref{protocol:RandQOT:step4}, where the sender publishes their measurement bases. At this point, the receiver can copy the sender's measurement bases, thus allowing the receiver to learn both $s_0$ and $s_1$.

Note that there is no step in \Cref{protocol:RandQOT} where an abort can occur, and so the probability of not aborting is one, $\Pr(Z=1)=1$. Hence, in the case that both parties are honest, we easily get $\Pr(Z=0)=0$ which satisfies the completeness condition of \Cref{randOT}.

Although we use Bell pairs in our protocol, and conjugate coding is used in the protocol in \cite{DFR+07}, the intuition behind perfect receiver-security in both protocols is the same, as is the strategy for proving it: the non-interactivity of both protocols means that a dishonest sender is unable to learn the receiver's choice bit; to prove perfect receiver-security, we consider the scenario where a dishonest sender executes the protocol with a receiver that has unbounded quantum memory and thus can compute $S'_0, S'_1$. Consequently, the only modification that must be made to the proof of Proposition 4.5 (perfect receiver-security) in \cite{DFR+07} is that we must take relation \eqref{equalBits} into account.

\begin{prop}
\label{baseReceiverSecurity}
\Cref{protocol:RandQOT} is perfectly receiver-secure.
\end{prop}

\begin{proof}
Since the probability of not aborting is one in \Cref{protocol:RandQOT}, we can drop the superscript $Z=1$ on our overall quantum state (since the superscript $Z=1$ denotes the state when the protocol does not abort; see \Cref{randOT}).

The ccq-state $\rho_{C Y \widetilde{S}}$ is defined by the experiment where $\widetilde{S}$ interacts with an honest memory-bounded $R$. Now, in a new Hilbert space, we define the ccccq-state $\hat{\rho}_{\hat{C} \hat{Y} \hat{S}'_0 \hat{S}'_1 \widetilde{S}}$ according to a different experiment.

In this different experiment, we let $\widetilde{S}$ interact with a receiver that has unbounded quantum memory. Let $V^\alpha,V^\beta$ be the strings that describe the random and independent choices of $v^\alpha, v^\beta$ which, together, describe which of the four Bell states are used in each round (see \cref{bellPairs}). Let~$A$ be the string the sender gets after measuring the $i$-th qubit in the basis~$x_i$ for $i=1, \ldots, n$. Define the string $W^{\alpha}$ in terms of $V^\alpha$ in the same way $w^\alpha$ is defined in terms of $v^\alpha$ in \Cref{protocol:RandQOT:step3}.

Now, the receiver waits to receive $x$ and then also measures the $i$-th qubit in the basis $x_i$ for $i = 1, \ldots, n$. Let $B$ be resulting string, and define the string $W^{\beta}$ in terms of $V^\beta$ in the same way $w^\beta$ is defined in terms of $v^\beta$ in \Cref{protocol:RandQOT:step2}.

Define $\hat{S}'_0 = f_0(A \oplus W^{\alpha} |_{I_0})$ and $\hat{S}'_1 = f_1(A \oplus W^{\alpha}|_{I_1})$. Note that $A \oplus W^{\alpha} = B \oplus W^{\beta}$. Sample $\hat{C}$ according to $P_C$ and set $\hat{Y} = \hat{S}'_C$. It follows by construction that $\text{Pr}[\hat{Y} \neq \hat{S}'_{\hat{C}}]=0$ and $\hat{\rho}_{\hat{C}}$ is independent of $\hat{\rho}_{\hat{S}'_0 \hat{S}'_1 \widetilde{S}}$.

It now remains to argue that $\hat{\rho}_{\hat{C} \hat{Y} \widetilde{S}} = \rho_{C Y \widetilde{S}}$ so that the corresponding $S'_0$ and $S'_1$ also exist in the original experiment. But, this is satisfied since the only difference between the two experiments is when and in what basis the qubits at position $i \in I_{1-C}$ are measured, which does not affect $\rho_{C Y \widetilde{S}}$ respectively $\hat{\rho}_{\hat{C} \hat{Y} \widetilde{S}}$.
\end{proof}

Now, for security against dishonest receivers, we restrict ourselves to receivers whose quantum storage is bounded during the execution of the protocol.

\begin{defn}[\cite{DFR+07}]
\label{setBoundedReceivers}
Let $\mathfrak{R}_\gamma$ denote the set of all possible quantum dishonest receivers $\widetilde{R}$ in \Cref{protocol:RandQOT} (or \Cref{protocol:DIOTsenderVerifier}) which have quantum memory of size at most $\gamma n$ when \Cref{protocol:RandQOT:step4} of \Cref{protocol:RandQOT} (or $\gamma n'$ when \Cref{protocol:DIOTsenderVerifier:step8} of \Cref{protocol:DIOTsenderVerifier}) is reached.
\end{defn}

The proof of $\varepsilon$-sender-security against $\mathfrak{R}_{\gamma}$, given below, is nearly identical to the proof of Theorem 4.6 ($\varepsilon$-sender-security) in \cite{DFR+07}. However, because we allow for the use of all four Bell states, we will need to show that
\begin{equation*}
    H^{\varepsilon}_{\infty}(K|X)=H^{\varepsilon}_{\infty}(A|X)
\end{equation*}
where $A$ is the random variable describing the outcome of the sender measuring their part of the quantum state in random basis $X$, and $K\coloneqq A \oplus W^\alpha$. Observe that if we only used Bell pairs of the form $\ket{\phi^{(0,0)}}$, then $W^\alpha=0$, and so our proof would fully reduce to the proof of Theorem 4.6 in \cite{DFR+07}, as expected.

\begin{prop}
\label{baseSenderSecurity}
\Cref{protocol:RandQOT} is $\varepsilon$-sender-secure against $\mathfrak{R}_\gamma$ for a negligible (in $n$) $\varepsilon$ if there exists a $k>0$ such that $\gamma n \leq n/4-2\ell-kn$.
\end{prop}

\begin{proof}
Since the probability of not aborting is one in \Cref{protocol:RandQOT}, we can drop the superscript $Z=1$ on our overall quantum state (since the superscript $Z=1$ denotes the state when the protocol does not abort; see \Cref{randOT}).

Now consider the quantum state in \Cref{protocol:RandQOT} after $\widetilde{R}$ has measured all but $\gamma n$ of their qubits. Let $V^\alpha,V^\beta$ be the random variables that describe the random and independent choices of $v^\alpha, v^\beta$ which, together, describe which of the four Bell states are used in each round (see \cref{bellPairs}). Define $W^\alpha$ in terms of $V^\alpha$ in the same way $w^\alpha$ is defined in terms of $v^\alpha$ in \Cref{protocol:RandQOT:step3}. Let $A$ be the random variable that describes the outcome of the sender measuring their part of the state in random basis $X$, and let $E$ be the random state that describes $\widetilde{R}$'s part of the state. Let $F_0$ and $F_1$ be the random variables that describe the random and independent choices of $f_0,f_1 \in \mathcal{F}$.

Choose $\lambda, \lambda', \kappa$ all positive, but small enough such that
\begin{equation}
\label{baseSenderSecurity:parameters}
    \gamma n \leq (1/4 - \lambda - 2\lambda' - \kappa)n-2\ell-1.
\end{equation}
From the uncertainty relation \Cref{entropicUncertainty}, we know that $H^\varepsilon_{\infty}(A | X) \geq (1/2-2\lambda)n$ for~$\varepsilon$ exponentially small in $n$. Let $A_r = A|_{I_r}$ and $W^{\alpha}_r = W^{\alpha}|_{I_r}$ where
\begin{align*}
    I_r = \{i: X_i = [\texttt{Computational, Hadamard}]_r\}.
\end{align*}
Define $K=A \oplus W^\alpha$ and let $K_r$ be $K|_{I_r}$ padded with zeros so that it will make sense to apply $F_r$. Now, to see that $H^{\varepsilon}_{\infty}(K|X)=H^{\varepsilon}_{\infty}(A|X)$, it suffices to observe that, for a single round,
\begin{align*}
    P_{AX\mathcal{E}}(a,x) &= P_{A\mathcal{E}|X}(a|x=0) \cdot P_{X}(x=0) + P_{A\mathcal{E}|X}(a|x=1) \cdot P_{X}(x=1),
\end{align*}
and,
\begin{align*}
    P_{A\mathcal{E}|X}(a|x=1) &= P_{KW^\alpha \mathcal{E}|X}(k \oplus w^{\alpha}=a|x=1)\\
    &= P_{K\mathcal{E}|X}(k = a \oplus w^\alpha|w^\alpha,x=1) \cdot P_{W^\alpha}(w^\alpha|x=1)\\
    &= \bigg(\sum^{2}_{j=1} P_{K\mathcal{E}|X}(k = a|x=1)\cdot\frac{1}{2}\bigg)\\
    &= P_{K\mathcal{E}|X}(k|x=1)\\
    &\\
    P_{A\mathcal{E}|X}(a|x=0) &= P_{KW^\alpha \mathcal{E}|X}(k \oplus w^{\alpha}=a|x=0)\\
    &= P_{K \mathcal{E}|X}(k=a \oplus w^\alpha|w^\alpha, x=0)\cdot P_{W^\alpha}(w^\alpha|x=0)\\
    &= P_{K\mathcal{E}|X}(k|x=0).
\end{align*}
As result of this, we have $H^{\varepsilon}_{\infty}(K|X)=H^{\varepsilon}_{\infty}(A|X)\geq (1/2-2\lambda)n$. Furthermore, we have $H^\varepsilon_{\infty}(K_0 K_1 | X) \geq (1/2-2\lambda)n$.

Therefore, by \Cref{entropySplit}, there exists a binary random variable $C'$ such that for \newline $\varepsilon'=2^{-\lambda'n}$, it holds that
\begin{equation}
\label{conditionC}
    H^{\varepsilon+\varepsilon'}_{\infty}(K_{1-C'}|X,C') \geq (1/4-\lambda-\lambda')n-1.
\end{equation}
It is clear that we can condition on the independent $F_{C'}$ and use the chain rule \Cref{chainRule} to obtain
\begin{align}
    H^{\varepsilon + 2\varepsilon'}_{\infty}(K_{1-C'}&|X F_{C'}(K_{C'}) F_{C'},C') \label{baseSenderSecurity:conditionFC} \\
    &\geq H^{\varepsilon+\varepsilon'}_{\infty}(K_{1-C'}F_{C'}(K_{C'})|XF_{C'}C')-H_0(F_{C'}(K_{C'})|XF_{C'}C')-\lambda'n \nonumber \\
    &\geq (1/4-\lambda-2\lambda')n-\ell-1 \nonumber\\
    &\geq \gamma n + \ell + \kappa n, \nonumber
\end{align}
by the choice of $\lambda, \lambda', \kappa$.

We write $S_0 = F_0(K_0)$ and $S_1 = F_1(K_1)$. Now, by setting $U=X S_{C'} F_{C'} C'$ and then applying \Cref{privacyAmp}, we get
\begin{align*}
    D(\rho_{S_{1-C'}F_{1-C'} X S_{C'} F_{C'} C' E}, & \, \tfrac{1}{2^{\ell}}\mathds{1} \otimes \rho_{F_{1-C'} X S_{C'} F_{C'} C' E})\\
    & \leq \frac{1}{2} 2^{-\frac{1}{2}(H^{\varepsilon+2\varepsilon'}_{\infty}(K_{1-C'}|X S_{C'} F_{C'} C')-\gamma n-\ell)}+(2\varepsilon+4\varepsilon')\\
    & \leq \frac{1}{2} 2^{-\frac{1}{2}\kappa n} + 2\varepsilon + 4\varepsilon'
\end{align*}
which is negligible, and thus we have $\varepsilon$-sender security.
\end{proof}


\subsection{Device-independence from computational assumptions}
\label{subsection:DIcomp}

We now discuss our second building block: a self-testing protocol for a single (quantum) device that relies on a computational assumption. We start in \Cref{subsubsection:OgSelfTest} with a discussion of the single-device self-testing protocol from \cite{MDCA21}, \Cref{protocol:testTwoVerifiers}, to outline how it works and to present the self-testing guarantee of the protocol, \Cref{selfTestGuarantee}. Then, in \Cref{subsubsection:ModSelfTest}, we present our single-device self-testing protocol, \Cref{protocol:testOneVerifier}, which is a slightly modified version of \Cref{protocol:testTwoVerifiers} that is appropriate for the setting where there are two parties that distrust each other, as in the case of OT.

\subsubsection{Original single-device self-testing protocol}
\label{subsubsection:OgSelfTest}

In short, the single-device self-testing protocol from \cite{MDCA21, MV21} uses computational assumptions to certify that a device, consisting of two components but connected by a quantum channel, has correctly prepared a quantum state and measured it according to the bases specified by the verifier(s). This protocol was first presented in \cite{MV21} as a protocol for a single verifier. Then, in \cite{MDCA21}, the protocol was stated in a form that involves two verifiers so that it could be used to formulate a device-independent quantum key distribution protocol that generates an information-theoretically secure key; this latter form of the protocol, presented as \Cref{protocol:testTwoVerifiers} below, will be our starting point.

The computational assumption made in this protocol is that the device cannot solve the Learning with Errors problem during the execution of the protocol; specifically, the underlying cryptographic primitive is an ENTCF family (see \Cref{subsection:ENTCF}). This computational assumption replaces the non-communication assumption that is necessary for typical self-testing protocols which rely on the violation of a Bell inequality. While non-local operations are then possible, there is an honest implementation which only requires local operations and EPR pairs that are distributed on-the-fly.

A brief description of this honest implementation is outlined throughout \Cref{protocol:testTwoVerifiers}, and is further described in the discussion of \Cref{protocol:testTwoVerifiers}'s winning condition. For a detailed description of the device's honest behaviour, the reader is referred to the appendix of \cite{MDCA21}; for further details, the reader is referred to \cite{MV21}.

\begin{breakablealgorithm}
\caption{Single-device self-testing with two verifiers}
\label{protocol:testTwoVerifiers}
\begin{algorithmic}[1]

\State
\label{protocol:testTwoVerifiers:step1}
Alice chooses a basis, called the {\em state basis}, $\theta^A \in \{\texttt{Computational, Hadamard}\}$ uniformly at random and generates a key $k^A$ together with a trapdoor $t^A$, where the generation procedure for $k^A$ and $t^A$ depends on $\theta^A$ and a security parameter $\eta$. Likewise, Bob generates $\theta^B, k^B$, and $t^B$. The keys are such that the device cannot efficiently compute the state bases $\theta^A, \theta^B$ from the keys $k^A, k^B$. Alice and Bob send the keys $k^A, k^B$ to the device.

\vspace{2mm}
\State
\label{protocol:testTwoVerifiers:step2}
Alice and Bob receive strings $c^A$ and $c^B$, respectively, from the device.
\newline
{\em Honest behaviour: prepare a product state $\ket{\psi^A}\ket{\psi^B}$, where $\ket{\psi^A}$ and $\ket{\psi^B}$ are, respectively, functions of $k^A$ and $k^B$. Measure part of $\ket{\psi^A}$ in the computational basis to obtain a string $c^A$ and send it to Alice, keeping the remainder of the state. Similarly, obtain $c^B$ from $\ket{\psi^B}$ and send it to Bob.}

\vspace{2mm}
\State
\label{protocol:testTwoVerifiers:step3}
Alice and Bob independently choose a {\em challenge type} $CT \in \{\texttt{a,b}\}$ uniformly at random and send them to the device.

\vspace{2mm}
\State
\label{protocol:testTwoVerifiers:step4}
If $CT=\texttt{a}$ for Alice (Bob):
\begin{enumerate}[i)]
    \item Alice (Bob) receives string $z^A$ ($z^B$) from the device. \newline {\em Honest behaviour: measure the remainder of the state $\ket{\psi^A}$ $(\ket{\psi^B})$ in the computational basis and send the resulting string $z^A$ $(z^B)$ to Alice (Bob).}
\end{enumerate}

\vspace{2mm}
\State
\label{protocol:testTwoVerifiers:step5}
If $CT=\texttt{b}$ for Alice (Bob):
\begin{enumerate}[i)]
    \item Alice (Bob) receives string $d^A$ ($d^B$) from the device. \newline {\em Honest behaviour: measure the remainder of the state $\ket{\psi^A}$ $(\ket{\psi^B})$, except for one qubit of the state, in the Hadamard basis and send the resulting string $d^A$ $(d^B)$ to Alice (Bob).}
    \item Alice (Bob) chooses a uniformly random measurement basis \newline $x \in \{\texttt{Computational, Hadamard}\}$ ($y \in \{\texttt{Computational, Hadamard}\}$) and sends it to the device.
    \item Alice (Bob) receives an answer bit $a$ ($b$) from the device. Alice (Bob) also receives the bit $h^A$ ($h^B$) from the device. \newline {\em Honest behaviour: when $CT=\texttt{b}$ for Alice and Bob, the remaining state has two qubits. In place of a controlled-Z operation, apply the circuit in \Cref{circuit} by using a single EPR pair that has been distributed on-the-fly, and return the bits $h^A$ and $h^B$ to Alice and Bob, respectively. Then apply a Hadamard gate on the second qubit. Then measure the first qubit in the basis $x$ and the second qubit in the basis $y$, obtaining outcomes $a,b \in \{0,1\}$, respectively. Send $a$ to Alice and $b$ to Bob.}
\end{enumerate}

\end{algorithmic}
\end{breakablealgorithm}

Before discussing the winning condition of \Cref{protocol:testTwoVerifiers}, we make a few remarks. First, we note that Alice's state basis $\theta^A$ determines which family her key $k^A$ will belong to: $k^A \in \mathcal{K_G}$ if $\theta^A=\texttt{Computational}$, and $k^A \in \mathcal{K_F}$ if $\theta^A = \texttt{Hadamard}$. Likewise for Bob.

We also note that the purpose of using the circuit in \Cref{circuit} at \Cref{protocol:testTwoVerifiers:step5}(iii) instead of a controlled-$Z$ gate is to replace the need for non-local operations. By using the circuit in \Cref{circuit}, an honest device can succeed in \Cref{protocol:testTwoVerifiers} with only local operations and EPR pairs that are distributed-on-the-fly. Note that if $\ket{\psi}_{AB}$ is the state just before the circuit is applied, then after the circuit we have
\begin{equation}
    \sigma_X^{h^A}\sigma_Z^{h^B} \otimes \sigma_X^{h^B}\sigma_Z^{h^A}CZ\ket{\psi}_{AB}
\end{equation}
where $CZ$ is the controlled-$Z$ gate. The operator $\sigma_X^{h^A}\sigma_Z^{h^B} \otimes \sigma_X^{h^B}\sigma_Z^{h^A}$ can be undone by having the device's components communicate the bits $h^A$ and $h^B$ with each other and then applying the appropriate local operations, though this is undesirable as the purpose in using the circuit in \Cref{circuit} is to allow an honest device to succeed without communication between its components. So, rather than having the device undo the operator $\sigma_X^{h^A}\sigma_Z^{h^B} \otimes \sigma_X^{h^B}\sigma_Z^{h^A}$, the checks that Alice and Bob need to perform are modified to account for the presence of the operator. It is also worth noting that by using the circuit in \Cref{circuit}, the actions of the component held by Alice and the component held by Bob are independent of each other. This means that the honest implementation described above can be extended to the case where Alice and Bob choose different challenge types $CT$, that is, each component acts according to the challenge type it received.

\begin{figure}[ht]
    \centering
\begin{tikzpicture}
    \node[inputOutput] at (0,0) (bla1) {\small $A$};
    \node[meter]       at (4,0) (meterA)  {};
    \node[inputOutput] at (5.2,0) (bla2) {\small $h^A$};

    \draw[thick,decorate,decoration={brace},xshift=-4pt,yshift=0pt]
    (0,-1.4) -- (0,-0.65) node [black,midway,xshift=-1cm] {$\ket{\text{EPR}}$};

    \node[inputOutput] at (0,-0.75) (bla5) {};
    \node[circle, fill=black, inner sep=0pt, minimum size=5pt] at (2.5,-0.75) {};
    \node[circle, thick, draw] at (2.5,0) {};
    \node[cross=4.2pt, rotate=45, thick] at (2.5,0) {};
    \node[inputOutput] at (4,-0.75) (bla6) {};

    \node[inputOutput] at (0,-1.3) (bla7) {};
    \node[operator]    at (1.5,-1.3) (op1) {$H$};
    \node[circle, fill=black, inner sep=0pt, minimum size=5pt] at (2.5,-1.3) {};
    \node[circle, thick, draw] at (2.5,-2.05) {};
    \node[cross=4.2pt, rotate=45, thick] at (2.5,-2.05) {};
    \node[inputOutput] at (4,-1.3) (bla8) {};

    \node[inputOutput] at (0,-2.05) (bla3) {\small $B$};
    \node[meter]       at (4,-2.05) (meterB)  {};
    \node[inputOutput] at (5.2,-2.05) (bla4) {\small $h^B$};

    \draw[thick] (bla1) -- (meterA);
    \draw[thick, double] (meterA) -- (bla2);

    \draw[thick] (bla3) -- (meterB);
    \draw[thick, double] (meterB) -- (bla4);

    \draw[thick] (bla5) -- (bla6);
    \draw[thick] (bla7) -- (op1);
    \draw[thick] (op1) -- (bla8);

    \draw[thick] (2.5,-1.3) -- (2.5,-2.05);
    \draw[thick] (2.5,-0.75) -- (2.5,0);
\end{tikzpicture}
    \caption{\label{circuit} Circuit from \cite{MDCA21} to replace controlled-$Z$ gate}
\end{figure}
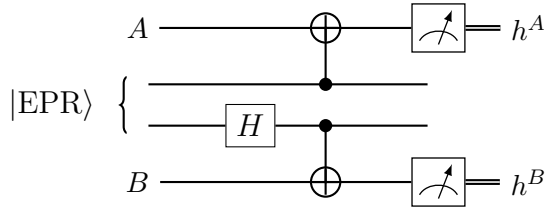

We now describe the winning condition of \Cref{protocol:testTwoVerifiers}. In doing this, we will briefly discuss the honest behaviour of the device that was outlined in \Cref{protocol:testTwoVerifiers}. We remind the reader that for a detailed description of the device's honest behaviour, and thus why the winning condition is as it is, the reader is referred to the appendix of \cite{MDCA21} and, for further details, to \cite{MV21}.

For the device to win a given round of \Cref{protocol:testTwoVerifiers}, several checks must be passed. If these checks pass, then Alice and Bob set a variable $W$ to \texttt{pass}; otherwise, $W=\texttt{fail}$.

\paragraph{If $CT=\texttt{a}$:}
For Alice, let $z_1^A$ be the first bit of $z^A$, and $z_r^A$ be the remainder of the string. Regardless of whether $k^A \in \mathcal{K_F}$ or $k^A \in \mathcal{K_G}$, Alice checks if $f_{k^A, z_1^A}(z_r^A)=c^A$. Likewise for Bob. If this check passes for both Alice and Bob, set $W=\texttt{pass}$.

\paragraph{If $CT=\texttt{b}$:}
For Alice, the honest behaviour of the device after \Cref{protocol:testTwoVerifiers:step5}(i) leaves the state $\ket{\psi^A}$ in one of two possible states. If $k^A \in \mathcal{K_F}$, then we have
\begin{equation}
\label{productForm}
    \ket{\psi_1^A} = \ket{0} + (-1)^{d^A \cdot (x_0^A \oplus x_1^A)}\ket{1}
\end{equation}
where $(x^A_0, x^A_1)$ is precisely the unique claw for the function pair $(f_{k^A,0}, f_{k^A,1})$ indexed by the key $k^A$, and the string $c^A$ satisfies
\begin{equation*}
    f_{k^A,0}(x^A_0) = f_{k^A,1}(x^A_1)=c^A.
\end{equation*}
If $k^A \in \mathcal{K_G}$, then we have
\begin{equation*}
    \ket{\psi_1^A} = \ket{\hat{b}^A}
\end{equation*}
where $\hat{b}^A \in \{0,1\}$ is such that
\begin{equation*}
    f_{k^A, \hat{b}^A}(\hat{x}^A)=c^A.
\end{equation*}
In \Cref{protocol:testTwoVerifiers:step5}(iii), just before measurement, the state held by the device is, up to a global phase,
\begin{equation*}
    \ket{\psi_2} = \sigma_X^{h^A}\sigma_Z^{h^B} \otimes \sigma_X^{h^B}\sigma_Z^{h^A}(\mathds{1} \otimes H)CZ\ket{\psi_1^A}\ket{\psi_1^B},
\end{equation*}
where the Hadamard gate has been commuted past the operator $\sigma_X^{h^A}\sigma_Z^{h^B} \otimes \sigma_X^{h^B}\sigma_Z^{h^A}$.

Unless $k^A \in \mathcal{K_F}$ and $k^B \in \mathcal{K_F}$, the state $\ket{\psi_2}$ is a product state. Together, Alice and Bob can determine precisely what product state the honest device has prepared. Indeed, with the trapdoor $t^A$, Alice can easily compute $(x_0^A, x_1^A)$ or $\hat{x}^A$ from $c^A$ (by property \ref{trapdoorEasy} of the ENTCF family), and likewise for Bob. Then, with $h^A, h^B, d^A$, and $d^B$, Alice and Bob have everything they need to determine $\ket{\psi_2}$. Knowing $\ket{\psi_2}$, Alice and Bob now determine what answers an honest device would have returned to their measurement basis questions $x$ and~$y$, and then check if the answers $a$ and $b$ returned by the device are the same; if they are, they set $W=\texttt{pass}$.

Now if $k^A \in \mathcal{K_F}$ and $k^B \in \mathcal{K_F}$ (\emph{i.e.}, $\theta^A = \theta^B = \texttt{Hadamard}$), then the state held by the device in \Cref{protocol:testTwoVerifiers:step5}(iii), just before measurement, turns out to be one of the four Bell states, up to a global phase,
\begin{equation*}
    \ket{\phi^{(d^A \cdot (x^A_0 \oplus \, x^A_1), \, d^B \cdot (x^B_0 \oplus \, x^B_1))}},
\end{equation*}
where
\begin{equation*}
    \ket{\phi^{(v^\alpha, \, v^\beta)}} = (\sigma^{v^\alpha}_Z \, \sigma^{v^\beta}_X \otimes \mathds{1}) \frac{\ket{00}+\ket{11}}{\sqrt{2}},
\end{equation*}
and so,
\begin{align*}
    v^\alpha &= d^A \cdot (x^A_0 \oplus x^A_1),\\
    v^\beta &= d^B \cdot (x^B_0 \oplus x^B_1).
\end{align*}
Recall that, with the trapdoor $t^A$, Alice can easily compute $(x_0^A, x_1^A)$, and likewise for Bob, and thus they can compute $v^{\alpha}$ and $v^{\beta}$. Recalling relation \eqref{equalBits}, Alice and Bob perform the following checks:
\begin{itemize}
    \item If $x=y=\texttt{Computational}$, check if $a\oplus b = d^B \cdot (x_0^B \oplus x_1^ B)$.
    \item If $x=y=\texttt{Hadamard}$, check if $a\oplus b = d^A \cdot (x_0^A \oplus x_1^A)$.
\end{itemize}
If one of the above checks pass, or if $x \neq y$, then set $W=\texttt{pass}$.

\paragraph{}
Now that the winning condition of \Cref{protocol:testTwoVerifiers} has been described, let us give the self-testing guarantee, stated as \Cref{selfTestGuarantee} below. Note that the guarantee is essentially stating that any computationally bounded device that wins in \Cref{protocol:testTwoVerifiers} must have performed single qubit measurements on a Bell state to obtain the results returned to the verifier. Recall that $\{Q^a_x\}_{a\in\{0,1\}}$ denotes the single-qubit measurement in the basis $x$ (see \Cref{subsection:notationBasics}), and note that for questions $x,y$, we denote the 4-outcome measurement used by the device to obtain answers $a,b$ by $\{P^{(a,b)}_{x,y}\}_{a,b \in \{0,1\}}$. We denote the state held by the device by $\sigma^{(v^\alpha,v^\beta)}$, where $v^\alpha, v^\beta$ are the bits that label which of the Bell states the device should have prepared, as in \cref{bellPairs}.

\begin{theorem}[\cite{MDCA21}]
\label{selfTestGuarantee}
Consider a device that wins \Cref{protocol:testTwoVerifiers} with probability $1-~\delta$ and make the LWE assumption. Let $\eta$ be the security parameter used in the protocol, $v^{\alpha}, v^{\beta} \in~\{0,1\}$ be the bits that label the Bell state as in \cref{bellPairs}, $\mathcal{H}$ be the device's physical Hilbert space, and $\mathcal{H}'$ be some ancillary Hilbert space. Then there exists an isometry $V: \mathcal{H} \rightarrow \mathbb{C}^4 \otimes \mathcal{H}'$ and some state $\zeta^{(v^\alpha,v^\beta)}_{\mathcal{H}'}$ such that, in the case $\theta^A=\theta^B=\texttt{Hadamard}$, the following holds:
\begin{align*}
    VP^{(a,b)}_{x,y}\sigma^{(v^\alpha,v^\beta)}P^{(a,b)}_{x,y}V^{\dagger}& \approx_{\mathcal{O}(\delta^r)+\textsf{negl}(\eta)}\\
    &\quad \bigg( (Q^a_x \otimes Q^b_y) \ket{\phi^{(v^\alpha,v^\beta)}}\bra{\phi^{(v^\alpha,v^\beta)}}(Q^a_x \otimes Q^b_y) \bigg) \otimes \zeta^{(v^\alpha,v^\beta)}_{\mathcal{H}'},
\end{align*}
where $r$ is some small constant arising in the proof.
\end{theorem}

It is worth mentioning why \Cref{selfTestGuarantee} only makes a statement about the case when $\theta^A = \theta^B = \texttt{Hadamard}$ and $CT = \texttt{b}$. In accordance with the honest implementation described in the discussion of the winning condition, rounds where $\theta^A = \theta^B = \texttt{Hadamard}$ and $CT = \texttt{b}$ are referred to as \texttt{Bell} rounds, while all other rounds are referred to as \texttt{Product} rounds.

The crucial point is that Alice and Bob will always know whether it is a \texttt{Bell} round or a \texttt{Product} round, but the computationally bounded device, which does not have access to $\theta^A$ and $\theta^B$, cannot determine what round it is in. This stems from the fact that the device is never given the trapdoor information $t^A, t^B$ and thus, by injective invariance of the ENTCF family (property \ref{injInv}), it is quantum-computationally hard for the device to determine which families the keys $k^A, k^B$ come from, and hence the type of round cannot be determined by the device.

Furthermore, Alice and Bob will always know precisely what Bell state or what product state the honest device should have prepared and, consequently, what answers should be returned in response to their measurement questions. So, to succeed and pass the checks of Alice and Bob, the device is forced to behave honestly. Thus, the self-testing guarantee only mentions \texttt{Bell} rounds.

Since \texttt{Bell} rounds are the ones where an honest device will prepare a Bell state, they are of primary interest. Recall that, in such a round, the state held by the device in \Cref{protocol:testTwoVerifiers:step5}(iii), just before measurement, is
\begin{equation*}
    \ket{\phi^{(d^A \cdot (x^A_0 \oplus \, x^A_1), \, d^B \cdot (x^B_0 \oplus \, x^B_1))}}.
\end{equation*}
As noted earlier, Alice and Bob can determine precisely which of the four Bell states have been prepared in a given \texttt{Bell} round. But, by the adaptive hardcore bit property of the ENTCF family (property \ref{adaptHardBit}), the device cannot efficiently compute $v^\alpha, v^\beta$ and hence cannot determine precisely what Bell state it has prepared.

\subsubsection{Modified single-device self-testing protocol}
\label{subsubsection:ModSelfTest}

It is important to note that \Cref{protocol:testTwoVerifiers} relies upon the verifiers, Alice and Bob, both behaving honestly. Specifically, after \Cref{protocol:testTwoVerifiers} has been executed, one of the two parties must publish their stored data so that the other can use it, along with their own stored data, to determine whether the device is behaving honestly or not. If, for instance, Bob is the one publishing his data for Alice to test the device, and he is dishonest, he could publish data that is different from what he received from the device, thereby giving Alice a false impression of the device's behavior.

This reliance on the honesty of Alice and Bob is a point of concern in the setting of~OT. Given this, we now present a variation of \Cref{protocol:testTwoVerifiers} where Alice is the sole verifier; Bob is still present in the protocol except that he is now being modelled as part of the device. The checks that must be performed are the same as the checks for \Cref{protocol:testTwoVerifiers}, except that they are all done by Alice now. The behaviour of an honest device in \Cref{protocol:testOneVerifier} is the same as the honest behaviour in \Cref{protocol:testTwoVerifiers}.

\begin{breakablealgorithm}
\caption{Single-device self-testing with a single verifier}
\label{protocol:testOneVerifier}
\begin{algorithmic}[1]

\State
\label{protocol:testOneVerifier:step1}
Alice chooses the state bases $\theta^A, \theta^B \in \{\texttt{Computational, Hadamard}\}$ uniformly at random and generates key-trapdoor pairs $(k^A,t^A), (k^B,t^B)$, where the generation procedure for $k^A$ and $t^A$ depends on $\theta^A$ and a security parameter $\eta$, and likewise for $k^B$ and $t^B$. Alice supplies Bob with $k^B$. Alice and Bob, respectively, then send the keys $k^A, k^B$ to the device.
\vspace{2mm}
\State
\label{protocol:testOneVerifier:step2}
Alice and Bob receive strings $c^A$ and $c^B$, respectively, from the device.

\vspace{2mm}
\State
\label{protocol:testOneVerifier:step3}
Alice chooses a {\em challenge type} $CT \in \{\texttt{a,b}\}$ uniformly at random and sends it to Bob. Alice and Bob then send $CT$ to each component of their device.

\vspace{2mm}
\State
\label{protocol:testOneVerifier:step4}
If $CT=\texttt{a}$:
\begin{enumerate}[i)]
    \item Alice and Bob receive strings $z^A$ and $z^B$, respectively, from the device.
\end{enumerate}

\vspace{2mm}
\State
\label{protocol:testOneVerifier:step5}
If $CT=\texttt{b}$:
\begin{enumerate}[i)]
    \item Alice and Bob receive strings $d^A$ and $d^B$, respectively, from the device.
    \item Alice chooses uniformly random measurement bases \newline $x,y \in \{\texttt{Computational, Hadamard}\}$ and sends $y$ to Bob. Alice and Bob then, respectively, send $x$ and $y$ to the device.
    \item Alice and Bob receive answer bits $a$ and $b$, respectively, from the device. Alice and Bob also receive bits $h^A$ and $h^B$, respectively, from the device.
\end{enumerate}

\end{algorithmic}
\end{breakablealgorithm}

The only role Bob has in \Cref{protocol:testOneVerifier} is to act as a relay between Alice and the component of the device held by Bob. It is clear that since Bob is not supplied with the state bases $\theta^A, \theta^B$ or the trapdoors $t^A, t^B$, any malicious behaviour from Bob can be folded into the device, and so without loss of generality, we can assume that Bob acts honestly in \Cref{protocol:testOneVerifier}. Then in order for Alice's checks to pass, the device must act honestly in \Cref{protocol:testOneVerifier} for the same reason that it must act honestly in \Cref{protocol:testTwoVerifiers}. Thus, \Cref{selfTestGuarantee} applies to \Cref{protocol:testOneVerifier} as well.

It should also be noted that to call \Cref{selfTestGuarantee}, we must know the probability $1-\delta$ with which the device wins \Cref{protocol:testOneVerifier}. Similar to the IID case in \cite{KW16}, the IID assumption enables Alice to estimate $\delta$ ahead of time; in our case, this can be done by having Bob temporarily give Alice his component of the device (so that Bob cannot influence the sample that Alice uses to estimate $\delta$).

Now suppose Alice uses $N$ rounds to estimate $\delta$. Let $F_i$ be a binary random variable for whether or not the device fails the $i$-th round. Then $F_1, \ldots, F_N$ are independent random variables, each of which is equal to 1 with probability $\delta$. If $F=F_1+\cdots+F_N$, then the expected value is $E(F)=N\delta$. Alice's estimate of $\delta$ is then $\delta' \coloneqq F/N$. If she wants her estimate to be within $\tau$ of $\delta$, then from the Chernoff bound, we have that,
\begin{align*}
    P(|\delta' - \delta|\geq\tau) &= P(|F-N \delta|\geq N\tau)\\
    &= P(|F-N\delta| \geq (\tau/\delta)N\delta)\\
    &\leq 2e^{-\frac{(\tau/\delta)^2N\delta}{3}}\\
    &= 2e^{-\frac{\tau^2N}{3\delta}}\\
    &\leq 2e^{-\frac{\tau^2N}{3}}
\end{align*}
Thus, $1-2e^{-\frac{\tau^2N}{3}} \leq P(\delta' \in (\delta-\tau, \delta+\tau))$. That is, Alice can improve her estimate $\delta'$ by choosing $\tau$ to be small and taking a large sample $N$.


\section{Device-Independent Oblivious Transfer}
\label{section:DIOT}

The goal of this section is to use the single-device self-testing protocol (\Cref{protocol:testOneVerifier}) to make the protocol for Rand 1-2 OT\textsuperscript{$\ell$} (\Cref{protocol:RandQOT}) device-independent. The result of this is \Cref{protocol:DIOTsenderVerifier}. It should be noted that \Cref{protocol:DIOTsenderVerifier} only considers the case where the sender is the verifier. Although it seems natural to require another protocol to allow the receiver to be the verifier, we find that such a protocol is entirely unnecessary due to the fact that we already have perfect-receiver-security for \Cref{protocol:DIOTsenderVerifier} (see \Cref{DIOTreceiverSecurity}).

Let us now describe \Cref{protocol:DIOTsenderVerifier}. The first five steps of \Cref{protocol:DIOTsenderVerifier} can be summarized as executing $n$ rounds of \Cref{protocol:testOneVerifier} (with the sender playing the role of Alice and the receiver playing the role of Bob), processing the data, and then checking if the device has behaved honestly for a subset of the rounds. In fact, these first five steps look very similar to the first six steps of the DIQKD protocol in \cite{MDCA21}. However, there are some key differences which we now discuss.

Firstly, the DIQKD protocol uses \Cref{protocol:testTwoVerifiers} for single-device self-testing while we use \Cref{protocol:testOneVerifier}. As noted earlier, the setting of OT is such that the sender and receiver do not necessarily trust each other. It is natural, then, that in making a device-independent version of \Cref{protocol:RandQOT}, we should not utilize \Cref{protocol:testTwoVerifiers} to verify the device, as this protocol relies on two cooperating parties. Instead, we should use \Cref{protocol:testOneVerifier}, which only requires one verifier.

Furthermore, it is assumed that the probability with which the device wins \Cref{protocol:testOneVerifier}, $1-\delta$, has been estimated by the sender prior to \Cref{protocol:DIOTsenderVerifier}, as discussed at the end of \Cref{subsubsection:ModSelfTest}. The result of this is that with probability at least $1-2e^{-\frac{\tau^2N}{3}}$, we have $\delta-\tau < \delta' < \delta+\tau$ and thus $\delta'-\tau < \delta$. We then use $\delta'-\tau$ as the threshold to check against in \Cref{protocol:DIOTsenderVerifier:step5}.

Secondly, we require that, with some probability, the receiver ignore the measurement basis question supplied by the sender and instead ask the device to measure in the basis specified by the choice bit. The purpose of this modification is to remedy the following problem. Since the sender is testing the device, they are supplying the receiver with all inputs for their component of the device. This gives the sender precise knowledge of what measurement basis the receiver is using for every round. Receiver-security is then compromised when the receiver tells the sender the indices of all rounds where they have retained, amongst the useable rounds, those where the measurement basis coincided with their choice bit; because the sender can immediately learn the choice bit from this. But, with our modification, we end up with a set of rounds $I$ where the sender is ignorant to the receiver's measurement basis questions, and thus ignorant to the choice bit. The sender remains ignorant to the choice bit even after the step where the sender tests the device's honesty; this is because when the receiver is required to send their stored data to the sender for the sake of testing, the receiver excludes data for rounds from $I$.

Now, we return to our description of \Cref{protocol:DIOTsenderVerifier}. In \Cref{protocol:DIOTsenderVerifier:step6}, the sender identifies a subset $\widetilde{I} \subseteq I$, which is the set of indices of all rounds in $I$ where the device has prepared a Bell pair and measured the sender's and receiver's half of the pair in the basis specified by each of them. Then for each round in $\widetilde{I}$, the sender publishes the trapdoor that corresponds to the receiver's key (the trapdoors are needed for the next step). Then, in \Cref{protocol:DIOTsenderVerifier:step7}, the sender and receiver correct their output in accordance with relation \eqref{equalBits}. At this point, we will have completed the first three steps of \Cref{protocol:RandQOT} in a device-independent manner.

The last two steps of \Cref{protocol:DIOTsenderVerifier}, \Cref{protocol:DIOTsenderVerifier:step8} and \Cref{protocol:DIOTsenderVerifier:step9}, are then identical to the last two steps of \Cref{protocol:RandQOT}, with slightly different notation.

\begin{breakablealgorithm}
\caption{DI Rand 1-2 OT\textsuperscript{$\ell$}}
\label{protocol:DIOTsenderVerifier}
\begin{algorithmic}[1]

\vspace{2mm}
\noindent {\bf Data generation:}
\vspace{2mm}
\State
\label{protocol:DIOTsenderVerifier:step1}
The sender and receiver execute $n$ rounds of \Cref{protocol:testOneVerifier} with the sender as Alice and the receiver as Bob, and with the following modification:

\vspace{2mm}
If $CT_i=\texttt{b}$, the receiver makes a uniformly random choice on whether to use the measurement basis question supplied by the sender or
\begin{equation*}
    y_i = [\texttt{Computational, Hadamard}]_c
\end{equation*}
where $c$ is the receiver's choice bit.\footnote{We thank Daochen Wang, Honghao Fu, and Qi Zhao for pointing out a correction related to the probability with which the receiver makes this measurement choice.} Let $I$ be the set of indices marking the rounds where this has been done.

\vspace{2mm}
\noindent For each round $i \in \{1, \ldots, n\}$, the receiver stores:
\begin{itemize}
    \item $c^B_i$
    \item $z^B_i$ if $CT_i=\texttt{a}$
    \item or $(d^B_i, y_i, b_i, h^B_i)$ if $CT_i=\texttt{b}$
\end{itemize}
The sender stores $\theta^A_i, \theta^B_i, (k^A_i, t^A_i), (k^B_i, t^B_i), c^A_i, CT_i \,$; and $z^A_i$ if $CT_i=\texttt{a}$ or $(d^A_i, x_i, a_i, h^A_i)$ and $y_i$ if $CT_i=\texttt{b}$.

\vspace{2mm}
\State
\label{protocol:DIOTsenderVerifier:step2}
For every $i \in \{1, \ldots, n\}$, the sender stores the variable $RT_i$ (round type), defined as follows:
\begin{itemize}
    \item if $CT_i = \texttt{b}$ and $\theta^A_i = \theta^B_i = \texttt{Hadamard}$, then $RT_i = \texttt{Bell}$
    \item else, set $RT_i = \texttt{Product}$
\end{itemize}

\vspace{2mm}
\State
\label{protocol:DIOTsenderVerifier:step3}
For every $i\in \{1,\ldots,n\}$, the sender chooses $T_i$, indicating a test round or generation round, as follows:
\begin{itemize}
    \item if $RT_i=\texttt{Bell}$, choose $T_i \in \{\texttt{Test, Generate}\}$ uniformly at random
    \item else, set $T_i=\texttt{Test}$.
\end{itemize}
The sender sends $(T_1,\ldots,T_n)$ to the receiver.

\vspace{5mm}
\noindent {\bf Testing:}
\vspace{2mm}
\State
\label{protocol:DIOTsenderVerifier:step4}
The receiver sends the set of indices $I$ to the sender. The receiver publishes their output for all $T_i=\texttt{Test}$ rounds where $i \notin I$. Using this published data, the sender sets a variable $W_i$ to \texttt{pass} if the checks for \Cref{protocol:testOneVerifier} are passed; otherwise, $W_i=\texttt{fail}$.

\vspace{2mm}
\State
\label{protocol:DIOTsenderVerifier:step5}
The sender computes the fraction of test rounds (for which the receiver has published data for) for which $W_i=\texttt{fail}$. If this exceeds the threshold $\delta'-\tau$ estimated by the sender prior to the protocol, then the protocol aborts.

\noindent {\bf Preparing data:}
\vspace{2mm}
\State
\label{protocol:DIOTsenderVerifier:step6}
Let $\widetilde{I} \coloneqq \{i : i \in I \text{ and } T_i = \texttt{Generate}\}$ and $n'=|\widetilde{I}|$. The sender publishes $\widetilde{I}$ and, for each $i \in \widetilde{I}$, the trapdoor $t^B_i$ that corresponds to the key $k^B_i$ that was given by the sender in the execution of \Cref{protocol:testOneVerifier}, \Cref{protocol:DIOTsenderVerifier:step1}.

\vspace{2mm}
\State
\label{protocol:DIOTsenderVerifier:step7}
For each $i \in \widetilde{I}$, the sender calculates $v^\alpha_i$ and defines $w^\alpha_i$ by
\begin{equation*}
    w^\alpha_i =
    \begin{cases}
        v^\alpha_i, &\text{if }  x_i = \texttt{Hadamard}\\
        0, &\text{if } x_i = \texttt{Computational}\\
    \end{cases}
\end{equation*}
and the receiver calculates $v^\beta_i$ and defines $w^\beta_i$ by
\begin{equation*}
    w^\beta_i =
    \begin{cases}
        0, &\text{if }  y_i = \texttt{Hadamard}\\
        v^\beta_i, &\text{if } y_i = \texttt{Computational}\\
    \end{cases}
\end{equation*}

\vspace{5mm}
\noindent {\bf Obtaining output:}
\vspace{2mm}
\State
\label{protocol:DIOTsenderVerifier:step8}
The sender picks two uniformly random hash functions $f_0, f_1 \in \mathcal{F}$, announces $f_0, f_1$ and $x_i$ for each $i \in \widetilde{I}$, and outputs $s_0=f_0(a\oplus w^\alpha|_{\widetilde{I}_0})$ and $s_1=f_1(a\oplus w^\alpha|_{\widetilde{I}_1})$, where \newline $\widetilde{I}_r \coloneqq \{i \in \widetilde{I} : x_i = [\texttt{Computational, Hadamard}]_r\}$.

\vspace{2mm}
\State
\label{protocol:DIOTsenderVerifier:step9}
Receiver outputs $s_c=f_c(b\oplus w^\beta|_{\widetilde{I}_c})$.

\end{algorithmic}
\end{breakablealgorithm}

Note that if \Cref{protocol:DIOTsenderVerifier:step5} is passed, then the fraction of failed test rounds does not exceed $\delta'-\tau$. Additionally, with probability at least $1-2e^{-\frac{\tau^2N}{3}}$, we have that the sender's estimate $1-\delta'$ of the device's winning probability $1-\delta$ satisfies $\delta'-\tau < \delta$. With the occurrence of these two events, the sender can use \Cref{selfTestGuarantee} and the IID assumption to say that for each $\texttt{Generate}$ round, one of the four Bell states has been prepared. This results in our statement of sender-security being conditioned on the high probability event that $\delta'-\tau<\delta$.

The rest of the protocol, which operates only on the rounds $i \in \widetilde{I}$, is then practically identical to \Cref{protocol:RandQOT}. The main difference here is that, for the relevant rounds, the sender is supplying the receiver with the trapdoor $t^B$ in \Cref{protocol:DIOTsenderVerifier:step6} to allow the receiver to compute $w^\beta$. Intuitively, this action does not give a dishonest receiver any advantage at this point, as the device has already been verified and the key-trapdoor pair only allows the receiver to learn one of the two bits that, collectively, indicates which of the four Bell states has been used in a given round.

Additionally, a dishonest sender also has no advantage in this device-independent setting when compared to \Cref{protocol:RandQOT}. Although the receiver must interact with the sender in \Cref{protocol:DIOTsenderVerifier}, while there was previously no need to in \Cref{protocol:RandQOT}, this interaction does not give the sender any information on the receiver's choice bit. To see this, observe that this interaction occurs in \Cref{protocol:DIOTsenderVerifier:step4} where the receiver publishes the set of indices $I$ and all outputs for $\texttt{Test}$ rounds so long as $i \notin I$. Thus, the output that the sender gets for $\texttt{Test}$ rounds says nothing about the receiver's actions since the input for these rounds was completely specified by the sender. It is only in the rounds where $i\in I$ that the receiver measures according to their choice bit, and for these rounds, only the set of indices $I$ is published, which says nothing of the actual choice bit.

It should also be noted why assumption \ref{leakageAssumption} is necessary. The use of \Cref{protocol:testOneVerifier} means that we can certify that the device has prepared a quantum state and measured it according to the prescribed measurement bases, while allowing arbitrary communication. Arbitrary communication poses a problem, though, to everlasting security in this setting where the sender and receiver do not trust each other. For instance, if the component held by the receiver leaked their measurement basis questions $y$, then the sender can immediately learn the receiver's choice bit by looking at what $y$ was in the rounds $i\in I$. Conversely, if the component held by the sender leaked their inputs and outputs, then a dishonest receiver could execute \Cref{protocol:DIOTsenderVerifier} honestly and still compromise sender-security. Indeed, suppose that, in executing the protocol honestly, the receiver obtained the string $s_0$ but stored the leaked inputs and outputs from the sender's component. To then learn $s_1=f_1(a\oplus w^{\alpha}|_{\widetilde{I}_1})$ after the protocol is over, the receiver must learn the sender's measurement outcomes $a$ and the bits $v^{\alpha}$ for the $\widetilde{I}_1$ rounds. The measurement outcomes $a$ would be amongst the leaked data, and so the task reduces to determining $v^{\alpha}$ from the following leaked data:
\begin{itemize}
    \item the key $k^A$ which indexes the function pair $(f_{k^A,0},f_{k^A,1})$
    \item the string $c^A$ which satisfies $f_{k^A,0}(x^A_0)=f_{k^A,1}(x^A_1)=c^A$
    \item the string $d^A$ which satisfies $v^{\alpha}=d^A \cdot (x^A_0 \oplus x^A_1)$
\end{itemize}
If the receiver can find the claw $(x^A_0, x^A_1)$, then they can compute $v^\alpha$. Finding the claw $(x^A_0,x^A_1)$ is quantum-computationally hard without access to the trapdoor $t^A$ (which was never given to the device), but with enough time and computational power, this could be done, and thus everlasting security for the sender is compromised. Consequently, assumption~\ref{leakageAssumption} is necessary in this context for everlasting security.

The proof of \Cref{baseReceiverSecurity} largely carries over to the proof of perfect receiver-security for \Cref{protocol:DIOTsenderVerifier}. As for sender-security for \Cref{protocol:DIOTsenderVerifier}, the proof is similar to the proof of \Cref{baseSenderSecurity}, though we will now have to use \Cref{selfTestGuarantee} and analyze the probability of not aborting. Note that the case where the sender, receiver, and the device behave honestly is analyzed in the proof of sender-security (see Case 1 of \Cref{DIsenderSecurity}).

\begin{prop}
\label{DIOTreceiverSecurity}
\Cref{protocol:DIOTsenderVerifier} is perfectly receiver-secure.
\end{prop}

\begin{proof}
The ccq-state $\rho^{Z=1}_{C Y \widetilde{S}}$ is defined by the experiment where $\widetilde{S}$ interacts with an honest memory-bounded $R$ and the protocol does not abort. Now, in a new Hilbert space, we define the ccccq-state $\hat{\rho}^{Z=1}_{\hat{C} \hat{Y} \hat{S}'_0 \hat{S}'_1 \widetilde{S}}$ according to a different experiment.

In this different experiment, we let $\widetilde{S}$ interact with a receiver that has unbounded quantum memory. Suppose the receiver has not actually input any measurement questions $y_i$ into their component of the device for rounds where $i \in I$. Let $A$ be the string the sender gets after inputting $x_i$ for $i \in \widetilde{I}$, and let $W^{\alpha}$ be the string where the $i$-th entry is defined as
\begin{equation*}
    W^\alpha_i =
    \begin{cases}
        v^\alpha_i, \quad &\text{if } x_i = \texttt{Hadamard}\\
        0, \quad &\text{if } x_i = \texttt{Computational}
    \end{cases}
\end{equation*}

Now, the receiver waits for \Cref{protocol:DIOTsenderVerifier:step8} to receive $x$ for the $i \in \widetilde{I}$ rounds. Let $B$ be the string the receiver gets after inputting $x$ for the $i \in \widetilde{I}$ rounds. By assumption \ref{leakageAssumption}, the sender is oblivious to the measurement basis questions the receiver has given to their component of the device, along with the answer bits returned to the receiver. Note that at this point the receiver will also have, from \Cref{protocol:DIOTsenderVerifier:step6}, $t^B_i$ for each $i \in \widetilde{I}$. The receiver uses $t^B_i$ to calculate $v^\beta_i$ for each $i \in \widetilde{I}$ and defines $W^\beta$ to be the string where the $i$-th entry is
\begin{equation*}
    W^\beta_i =
    \begin{cases}
        0, \quad &\text{if } x_i = \texttt{Hadamard}\\
        v^{\beta}_i, \quad &\text{if } x_i = \texttt{Computational}
    \end{cases}
\end{equation*}

Define $\hat{S}'_0 = f_0(A \oplus W^{\alpha} |_{\widetilde{I}_0})$ and $\hat{S}'_1 = f_1(A \oplus W^{\alpha}|_{\widetilde{I}_1})$. Note that $A \oplus W^{\alpha} = B \oplus W^{\beta}$. Sample $\hat{C}$ according to $P_C$ and set $\hat{Y} = \hat{S}'_C$. It follows by construction that $\text{Pr}[\hat{Y} \neq \hat{S}'_{\hat{C}}]=0$ and $\hat{\rho}^{Z=1}_{\hat{C}}$ is independent of $\hat{\rho}^{Z=1}_{\hat{S}'_0 \hat{S}'_1 \widetilde{S}}$.

It now remains to argue that,
\begin{equation*}
    \hat{\rho}^{Z=1}_{\hat{C} \hat{Y} \widetilde{S}} = \rho^{Z=1}_{C Y \widetilde{S}}
\end{equation*}
so that the corresponding $S'_0$ and $S'_1$ also exist in the original experiment. But, this is satisfied since the only difference between the two experiments is when and what $x_i$ the receiver inputs for $i \in \widetilde{I}_{1-C}$ rounds, which does not affect $\rho^{Z=1}_{C Y \widetilde{S}}$ respectively $\hat{\rho}^{Z=1}_{\hat{C} \hat{Y} \widetilde{S}}$.
\end{proof}

For the following proposition, recall \Cref{setBoundedReceivers} which defines $\mathfrak{R}_\gamma$ as the set of all possible quantum dishonest receivers in \Cref{protocol:DIOTsenderVerifier} which have quantum memory of size at most $\gamma n'$ when \Cref{protocol:DIOTsenderVerifier:step8} of \Cref{protocol:DIOTsenderVerifier} is reached.

Also note that if we were not dealing with finite statistics, then the sender's initial estimate $1-\delta'$ of the device's winning probability $1-\delta$ could be done with arbitrary precision and so the following proposition would no longer be conditional on $\delta'-\tau<\delta$.

\begin{prop}
\label{DIsenderSecurity}
\Cref{protocol:DIOTsenderVerifier} is $(\varepsilon+(\delta'-\tau)^r)$-sender-secure against $\mathfrak{R}_\gamma$ for a negligible (in $n'$ and~$\eta$) $\varepsilon + (\delta'-\tau)^r$ if $\delta'-\tau<\delta$ (which occurs with probability at least $1-2e^{-\frac{\tau^2N}{3}}$) and if there exists a $k>0$ such that $\gamma n' \leq n'/4 - 2\ell - kn'$, where $n' = |\widetilde{I}|$, $\eta$ is the security parameter used in $\Cref{protocol:testOneVerifier}$, $(\delta'-\tau)$ and $N$ is arising from the sender's estimate of $\delta$, and $r$ is from the application of \Cref{selfTestGuarantee}.
\end{prop}

\begin{proof}

We consider the different cases regarding the probability of not aborting. Let $\delta_F$ denote the fraction of failed test rounds in \Cref{protocol:DIOTsenderVerifier}. If $\delta_F$ exceeds the threshold $\delta'-\tau$, \Cref{protocol:DIOTsenderVerifier} aborts. Given this, let $Z$ be the binary random variable which describes whether \Cref{protocol:DIOTsenderVerifier} aborts or not. That is,
\begin{equation}
Z \coloneqq
\begin{cases}
    1, &\text{ if } \delta_F \leq \delta'-\tau\\
    0, &\text{ if } \delta_F > \delta'-\tau
\end{cases}
\end{equation}

\paragraph{Case 1, honest behaviour:}
When both parties and the device behave honestly, the fraction of failed test rounds $\delta_F$ is small; by this, we mean that
\begin{equation*}
    \Pr(\delta_F \leq \delta'-\tau) \geq 1-(\varepsilon + (\delta'-\tau)^r).
\end{equation*}
Since $E(Z)=\Pr(\delta_F \leq \delta'-\tau)$, we then have $1-E(Z) \leq (\varepsilon + (\delta'-\tau)^r)$. We can now use the Chernoff bound to show that the probability of aborting is small.
\begin{align}
    \Pr(Z=0) &=P_Z(z \leq 0) \nonumber\\
    & =P_Z(-z \geq 0) \nonumber\\
    & \leq e^{-t} E(Z) + (1-E(Z)) \nonumber\\
    & \leq e^{-t} E(Z) + (\varepsilon + (\delta'-\tau)^r), \quad \forall t\geq 0. \label{honestAbort}
\end{align}
Then taking the limit $t \rightarrow \infty$ in \cref{honestAbort}, we see that the probability of aborting is small, $\Pr(Z=~0) \leq~(\varepsilon +~(\delta'-\tau)^r)$, which satisfies the completeness condition of \Cref{randOT}.

\paragraph{Case 2, dishonest behaviour and large $\delta_F$:}
We now show that when the behaviour is dishonest and the fraction of failed test rounds $\delta_F$ is large, in the sense that
\begin{equation*}
    \Pr(\delta_F \leq \delta'-\tau) \leq (\varepsilon + (\delta'-\tau)^r),
\end{equation*}
the probability of not aborting is small. For this, we again use the Chernoff bound,
\begin{align}
    \Pr(Z=1) &=P_Z(z \geq 1) \nonumber\\
    & \leq E(Z) + e^{-t}(1-E(Z)) \nonumber\\
    & \leq (\varepsilon + (\delta'-\tau)^r) + e^{-t}(1-E(Z)), \quad \forall t\geq 0. \label{dishonestLargeAbort}
\end{align}
Then taking the limit $t \rightarrow \infty$ in \cref{dishonestLargeAbort}, the probability of not aborting is small,
\begin{equation*}
    \Pr(Z=1) \leq~(\varepsilon +~(\delta'-~\tau)^r).
\end{equation*}
Thus, the overall expression in \cref{senderSecurity} is satisfied.

\paragraph{Case 3, dishonest behaviour and small $\delta_F$:}
In this case, we bound the overall expression in \cref{senderSecurity}. To do this, we first consider a single round $i \in \widetilde{I}$. For this single round, we make a slight abuse of notation by letting $v^\alpha, v^\beta$ be, respectively, the bits computed by the sender and the receiver in \Cref{protocol:DIOTsenderVerifier:step7}, and letting $X, Y$ and $A, B$ be the classical random variables describing the sender's and receiver's questions and answers, respectively. Let $\sigma^{(v^\alpha,v^\beta)}$ be the joint state of the device in \Cref{protocol:DIOTsenderVerifier:step1} right before the device performs the measurements $P^{(a,b)}_{x,y}$. Then the state after \Cref{protocol:DIOTsenderVerifier:step1}~is
\begin{equation*}
    \sum_{x,y,a,b} P_X(x) P_Y(y) \tr[P^{(a,b)}_{x,y} \sigma^{(v^\alpha,v^\beta)} P^{(a,b)}_{x,y}] \otimes \ket{x, y, a, b}\bra{x, y, a, b}_{XYAB}.
\end{equation*}
Now we consider the event $\delta'-\tau < \delta$, which, from the end of \Cref{subsubsection:ModSelfTest}, occurs with probability at least $1-2e^{-\frac{\tau^2N}{3}}$. Observe that if the protocol did not abort at \Cref{protocol:DIOTsenderVerifier:step5}, then $\delta_F \leq \delta'-\tau$, and hence, $1-\delta < 1-(\delta'-\tau) \leq 1-\delta_F$, meaning that the winning condition of \Cref{protocol:testOneVerifier} is satisfied with probability at least $1-(\delta'-\tau)$ in the \texttt{Test} rounds; but after \Cref{protocol:DIOTsenderVerifier:step1}, it has not yet been decided whether a round will be a \texttt{Test} round or a \texttt{Generate} round, and so, we can use the IID assumption to apply \Cref{selfTestGuarantee} to \texttt{Generate} rounds. Using \Cref{selfTestGuarantee} in our $i \in \widetilde{I}$ round, the continuous and cyclical properties of the trace, and that $V^\dagger V = \mathds{1}$, we find that the state for this round $\rho^{Z=1}_{XYAB}$ must be within trace distance $\mathcal{O}((\delta'-\tau)^r)+\textsf{negl}(\eta)$ of the ideal state
\begin{equation*}
    \xi_{XYAB} = \sum_{v^\alpha,v^\beta,x,y,a,b} P_X(x) P_Y(y) \bra{\phi^{(v^\alpha,v^\beta)}} Q^a_x \otimes Q^b_y \ket{\phi^{(v^\alpha,v^\beta)}} \otimes \ket{x, y, a, b}\bra{x, y, a, b}_{XYAB}.
\end{equation*}
This confirms that in each round $i \in \widetilde{I}$, the device has prepared a Bell state and measured it according to the measurement basis choices of the sender and the receiver.

Now, let us consider the state $\rho^{Z=1}_{XAOE}$ in the scenario where the protocol has not aborted and $\widetilde{R}$ has measured all but $\gamma n'$ of their qubits from the $\widetilde{I}$ rounds, where $n'=|\widetilde{I}|$. For the rest of the proof, we make a slight abuse of notation by dropping the $Z=1$ superscript on the state for the sake of readability. Now, since we are making the IID assumption, this state is an $n'$-fold tensor product and the following are $n'$-tuples with each entry representing one of the $i\in\widetilde{I}$ rounds:
\begin{itemize}
    \item $X$ is the classical random variable describing the random choice of bases of the sender.
    \item $A$ is the classical random variable describing the sender's results after measuring their part of the state in the random bases $X$.
    \item $O$ is the classical random variable which contains the random choices of the key-trapdoor pairs $(k^B, t^B)$ and the measurement basis questions $y$ supplied by the sender. It will also contain any other information that the sender's component of the device may have leaked to the receiver; this may include $k^A, c^A, d^A$, but by assumption \ref{leakageAssumption}, it cannot include the sender's measurement basis questions $x$ or their answer bits $a$.
    \item $E$ is the random state that describes $\widetilde{R}$'s part of the state.
\end{itemize}
The state $\rho_{XAOE}$ then satisfies
\begin{equation*}
    D(\rho_{XAOE}, \xi_{XAOE}) \leq n'\Big( \mathcal{O}((\delta'-\tau)^r)+\textsf{negl}(\eta)\Big) = \mathcal{O}((\delta'-\tau)^r)+\textsf{negl}(\eta)
\end{equation*}
where $\xi_{XAOE}$ is the ideal state.

Analyzing the ideal state, the proof is now similar to the proof of \Cref{baseSenderSecurity}. We start by lower bounding the smooth min-entropy for the state that is not conditioned on the event that the protocol does not abort.

At \cref{baseSenderSecurity:parameters}, we choose $\lambda, \lambda', \kappa$ such that
\begin{equation}
\label{senderSecurity:parameters}
    \gamma n' \leq (1/4 - \lambda - 2\lambda' - \kappa)n'-2\ell-1.
\end{equation}

Then at \cref{conditionC}, we have
\begin{equation*}
    H^{\varepsilon + \varepsilon'}_{\infty}(K_{1-C'}|X,C') \geq (1/4-\lambda-\lambda')n'-1,
\end{equation*}
where $\varepsilon$ is exponentially small in $n'$ and comes from \Cref{entropicUncertainty}, and $\varepsilon'=2^{-\lambda'n'}$.

Then at \cref{baseSenderSecurity:conditionFC}, in addition to conditioning on $F_{C'}$, we also condition on the random variable $O$ because it too is independent. It is easy to see this for the measurement questions~$y$ supplied by the sender since they are generated uniformly randomly. Regarding the key-trapdoor pairs, observe that knowledge of $(k^B,t^B)$ makes properties \ref{adaptHardBit} and \ref{injInv} of the ENTCF family computationally easy (because of property \ref{trapdoorEasy}). That is,
\begin{itemize}
    \item With $(k^B,t^B)$, the receiver can overcome the injective invariance property of the ENTCF family (property \ref{injInv}) and determine what $\theta^B$ is. However, the sender only publishes $(k^B,t^B)$ for \texttt{Bell} rounds, as these are the only types of rounds that are useable for accomplishing Rand 1-2 OT\textsuperscript{$\ell$}, and so the ability to overcome the injective invariance property and learn $\theta^B$ is redundant at this point.

    \item With $(k^B,t^B)$, property \ref{adaptHardBit} of the ENTCF family becomes computationally easy. That is, the unique claw $(x^B_0, x^B_1)$ can be calculated with the function pair indexed by the key $k^B$, the trapdoor $t^B$, and the string $c^B$. This means the receiver can determine the precise form of the state $\ket{\psi_1^B}$ (this is the analogous state of $\ket{\psi_1^A}$ in \cref{productForm}), but $\ket{\psi_1^B}$ is independent of $\ket{\psi_1^A}$ until the entangling operation. After the entangling operation, knowledge of $\ket{\psi_1^B}$ is equivalent to learning $v^\beta$, but this is necessary and accounted for with the random variable $K$ in the proof of \Cref{baseSenderSecurity}.
\end{itemize}
As for information leaked from the sender's component of the device, knowledge of $k^A, c^A, d^A$ makes it possible for the receiver to compute the bit $v^\alpha$ (see the discussion immediately before \Cref{DIOTreceiverSecurity}), though they cannot do this efficiently since they do not have access to the sender's trapdoor $t^A$. Given that the receiver will also know $v^\beta$, the receiver could eventually learn precisely what Bell state was used in each of the $\widetilde{I}$ rounds. To then learn the other string $s_{1-c}$, the receiver needs $a|_{\widetilde{I}_{1-c}}$ which appears uniformly random to the receiver since the sender and receiver chose different measurement bases for the $\widetilde{I}_{1-c}$ rounds. Thus, the receiver can do no better than guess, correctly, the other string $s_{1-c}$ with probability $1/2^\ell$.

So, conditioning on $O$ as well, and using the chain rule \Cref{chainRule} as in \cref{baseSenderSecurity:conditionFC}, we obtain
\begin{equation*}
    H^{\varepsilon + 2\varepsilon'}_{\infty}(K_{1-C'}|X F_{C'}(K_{C'}) F_{C'} O,C') \geq \gamma n' + \ell + \kappa n'.
\end{equation*}

We now consider the smooth min-entropy of the state conditioned on not aborting (we use $\{Z=1\}$ to denote the event of not aborting), which is lower bounded by the unconditioned smooth min-entropy (see Lemma 10 of \cite{TL17}). Conditioning the min-entropy on this event decreases the min-entropy by at most $\log(1/\Pr(Z=1))$. So,
\begin{align*}
    H^{\varepsilon + 2\varepsilon'}_{\infty}(&K_{1-C'}|X F_{C'}(K_{C'}) F_{C'} O,C')_{\{Z=1\}}\\
    &\geq H^{\varepsilon + 2\varepsilon'}_{\infty}(K_{1-C'}|X F_{C'}(K_{C'}) F_{C'} O,C') - 2\log\bigg(\frac{1}{\Pr(Z=1)}\bigg)\\
    &\geq \gamma n' + \ell + \kappa n' - 2\log\bigg(\frac{1}{\Pr(Z=1)}\bigg)
\end{align*}

Additionally, let $\varepsilon'' \coloneqq (\varepsilon + 2\varepsilon')/\Pr(Z=1)$. Then,
\begin{align*}
    H^{\varepsilon''}_{\infty}(K_{1-C'}|X F_{C'}(K_{C'}) F_{C'} O,C')_{\{Z=1\}}&\geq H^{\varepsilon + 2\varepsilon'}_{\infty}(K_{1-C'}|X F_{C'}(K_{C'}) F_{C'} O,C')_{\{Z=1\}}\\
    &\geq \gamma n' + \ell + \kappa n' - 2\log\bigg(\frac{1}{\Pr(Z=1)}\bigg)
\end{align*}

Letting $S_0 = F_0(K_0)$ and $S_1 = F_1(K_1)$, and setting $U= X S_{C'} F_{C'}O C'$, we get, after applying \Cref{privacyAmp},
\begin{align*}
    D(\xi_{S_{1-C'}F_{1-C'} X S_{C'} F_{C'} O C' E}, & \, \tfrac{1}{2^{\ell}}\mathds{1} \otimes \xi_{F_{1-C'} X S_{C'} F_{C'} O C' E})\\
    & \leq \frac{1}{2} 2^{-\frac{1}{2}(H^{\varepsilon''}_{\infty}(K_{1-C'}| X S_{C'} F_{C'} O C')_{\{Z=1\}}-\gamma n'-\ell)}+2\varepsilon''\\
    & \leq \bigg(\frac{1}{2} 2^{-\frac{1}{2}\kappa n'} + 2\varepsilon + 4\varepsilon'\bigg)\frac{1}{\Pr(Z=1)}.
\end{align*}
Now, using the triangle inequality twice,
\begin{align*}
    D(\rho_{S_{1-C'} F_{1-C'} U E}, \,& \tfrac{1}{2^{\ell}}\mathds{1} \otimes \rho_{F_{1-C'} U E}) \\
    &\leq D(\rho_{S_{1-C'}F_{1-C'} U E}, \xi_{S_{1-C'}F_{1-C'} U E}) + D(\xi_{S_{1-C'}F_{1-C'} U E}, \, \tfrac{1}{2^{\ell}}\mathds{1} \otimes \rho_{F_{1-C'} U E})\\
    &\leq \mathcal{O}((\delta'-\tau)^r) + \textsf{negl}(\eta) + D(\xi_{S_{1-C'}F_{1-C'} U E}, \, \tfrac{1}{2^{\ell}}\mathds{1} \otimes \xi_{F_{1-C'} U E})\nonumber\\
    & \qquad\qquad\qquad\qquad\qquad\qquad + D(\tfrac{1}{2^{\ell}}\mathds{1} \otimes \xi_{F_{1-C'} U E},\tfrac{1}{2^{\ell}}\mathds{1} \otimes \rho_{F_{1-C'} U E})\\
    &\leq 2(\mathcal{O}((\delta'-\tau)^r)+\textsf{negl}(\eta)) + \bigg(\frac{1}{2} 2^{-\frac{1}{2}\kappa n'} + 2\varepsilon + 4\varepsilon'\bigg)\frac{1}{\Pr(Z=1)}
\end{align*}
Thus,
\begin{align*}
    \Pr(Z=1) D(\rho_{S_{1-C'} F_{1-C'} U E}, \,& \tfrac{1}{2^{\ell}}\mathds{1} \otimes \rho_{F_{1-C'} U E})\\
    & \leq 2 (\mathcal{O}((\delta'-\tau)^r)+\textsf{negl}(\eta)) + \frac{1}{2} 2^{-\frac{1}{2}\kappa n'} + (2\varepsilon + 4\varepsilon') \qedhere
\end{align*}

\end{proof}

\fi


\bibliographystyle{alphaarxiv.bst}
\bibliography{full.bib,quantum.bib}

\end{document}